\documentclass[aps,pra,reprint,floatfix,unsortedaddress,superscriptaddress]{revtex4-2}

\usepackage{times}
\usepackage{amsmath,amsfonts,amssymb,amsthm}
\usepackage{braket}
\usepackage{graphicx}
\usepackage[psdextra]{hyperref}
\hypersetup{
    colorlinks=true,
    linkcolor=red,
    anchorcolor=blue,
    citecolor=blue,
    urlcolor=red
}

\newtheorem{theorem}{Theorem}[section]
\newtheorem{lemma}[theorem]{Lemma}
\newtheorem{corollary}[theorem]{Corollary}
\newtheorem{proposition}[theorem]{Proposition}

\newcommand{\casql}{Laboratory of Quantum Information, University of Science and Technology of China, Hefei, Anhui, 230026, China}
\newcommand{\aihf}{Institute of Artificial Intelligence, Hefei Comprehensive National Science Center, Hefei, Anhui, 230088, China}
\newcommand{\origin}{Origin Quantum Computing Technology (Hefei) Co., Ltd., Hefei, Anhui, 230026, China}

\newcommand{\norm}[1]{\left\lVert#1\right\rVert}
\newcommand{\abs}[1]{\left\lvert#1\right\rvert}
\newcommand{\Ocal}{\mathcal{O}}

\begin{document}

\title{Poisson-Compiled Quantum Singular Value Transformation for Power-Exponential Dissipation}

\author{Chao Wang}
\email{wangchao2@originqc.com}
\affiliation{\origin}
\author{Xi-Ning Zhuang}
\affiliation{\origin}
\affiliation{\casql}
\author{Menghan Dou}
\affiliation{\origin}
\author{Zhao-Yun Chen}
\email{chenzhaoyun@iai.ustc.edu.cn}
\affiliation{\aihf}
\author{Guo-Ping Guo}
\affiliation{\origin}
\affiliation{\aihf}
\affiliation{\casql}

\date{\today}

\begin{abstract}
We study quantum implementations of the contraction
$\exp(-T H^\alpha)$ for $H=H^\dagger\succeq0$ and $\alpha>0$.
Poisson summation generates an exact target--alias--tail decomposition,
whose Fourier samples are compiled classically into one Chebyshev
polynomial.  The quantum circuit therefore uses a polynomial eigenvalue
transformation rather than a frequency linear combination of unitaries.
We analyze two access models: block encodings of
$H/\norm{H}$ and of the shifted signal $2H/\norm{H}-I$.  In ordinary
single-sequence quantum singular value transformation (QSVT), parity
forces the first model to use an even
extension, which is entire only for even positive integers.  An exact
quadratic lift in the shifted model makes every positive integer entire
and changes the fixed-scale approximation error for noninteger powers
from order $d^{-\alpha}$ to $d^{-2\alpha}$ relative to that direct even
construction.
The resulting degree bounds are tight, within the stated access and
parity classes, in the separately stated large-scale, fixed-error and
fixed-scale, high-precision limits.
Preparing a normalized output introduces the input-to-output norm ratio
$u_r$ through both error allocation and amplitude amplification.
Nearest-neighbor Laplacians provide a unit-normalized realization of the
shifted signal.  For time-independent non-Hermitian dynamics, the
Fourier--Weyl representation of the entire surrogate used by optimal
linear combination of Hamiltonian simulation (LCHS) admits an exact
Weyl--Poisson alias identity and remains compatible with sinh--sinh
quadrature without a commutativity assumption.  This reformulation does
not improve the optimal LCHS query order.  In amplitude--phase
separation (APS), the dissipative semigroup law allows the same
polynomial construction to realize a uniformly controlled family of
contractions.  Square-root access recovers the usual square-root degree,
while higher root access gives a conditional beyond-square-root
extension.
\end{abstract}

\maketitle

\section{Introduction}
\label{sec:introduction}

Quantum circuits directly implement unitary dynamics, whereas heat
flow, higher-order and fractional diffusion, and effective
non-Hermitian dynamics are generated by contractions
~\cite{Linden2022,PhysRevA.108.032603,METZLER20001,LISCHKE2020109009}.
The techniques developed for Hamiltonian simulation provide much of
the required unitary machinery
~\cite{doi:10.1126/science.273.5278.1073,RevModPhys.86.153,
doi:10.1073/pnas.1801723115,Altman2021qsim,daley2022practical}, but a
dissipative propagator also requires a non-unitary approximation and
probabilistic preparation of the normalized output.

For $T>0$, we consider the homogeneous, time-independent problem
\begin{equation}
    u_T=\mathrm e^{-T H^\alpha}u_0,
    \qquad
    H=H^\dagger\succeq0,
    \qquad
    \alpha>0.
    \label{eq:intro_problem}
\end{equation}
Two access models are kept separate: a block encoding of
$H/\norm{H}$ and a unit-normalized block encoding of
$2H/\norm{H}-I$.  The latter is an additional structural assumption,
not a consequence of a generic block encoding.  In a sparse-matrix
implementation the two oracles can nevertheless share the row-location
data; the shifted entries require only a factor-of-two rescaling and a
diagonal subtraction~\cite{Berry2007,Sunderhauf2024blockencoding}.
The periodic-grid Laplacian in
Sec.~\ref{subsec:pde} also preserves unit normalization.
Neither access model requires a fractional power of $H$.  We write the
state-preparation overhead as
\begin{equation}
    u_r:=
    \frac{\norm{u_0}}{\norm{u_T}}.
    \label{eq:intro_output_ratio}
\end{equation}
The same notation is used below when $u_T$ is the state after a more
general time-independent dissipative evolution.

Quantum signal processing (QSP) and quantum singular value transformation
(QSVT) implement polynomial transformations of block-encoded
operators and underlie nearly optimal Hamiltonian simulation and
linear-system
algorithms~\cite{PhysRevLett.118.010501,Gilyen2019qsvt,
PhysRevLett.114.090502,
PhysRevLett.103.150502,doi:10.1137/16M1087072,
Low2026quantumlinearsystem}.
A direct Fourier representation of a matrix function can instead be
implemented by a linear combination of unitaries
(LCU)~\cite{Childs2012lcu}.  Its sampled frequencies become coherent
circuit branches, and the coefficient one-norm enters the success
amplitude.  In the construction below, Fourier analysis is used only
to generate a polynomial classically; one polynomial-transform circuit
implements the resulting transformation.

Related structured constructions use a Gaussian linear combination of
Hamiltonian simulation (LCHS) or the Kannai transmutation formula to
obtain square-root dissipative dependence under factorized or matched
access assumptions~\cite{Kannai01011977,
jin2026transmutationbasedquantumsimulation,
kharazi2026sublineartimequantumalgorithmhighdimensional}.  Here the
question is how the power $\alpha$ and the encoded signal jointly set
the polynomial degree, and how the resulting transformation can be
used in time-independent non-unitary simulation.

The mathematical input is the Poisson summation
formula~\cite{Miller2004,stein1971introduction,
katznelson2004introduction}.  Applied to
$\mathrm e^{-T\norm{H}^{\alpha}\abs{x}^{\alpha}}$, it expresses uniform
samples of the inverse Fourier transform as the target plus shifted
aliases.  Kernel symmetry removes the sine terms, and the
Jacobi--Anger expansion converts the remaining cosines into Chebyshev
coefficients.  These coefficients are summed before the circuit is
built, so the Poisson sample count belongs to classical preprocessing.
The exact target--alias--tail identity controls coefficient generation;
the quantum operation itself is standard QSVT and remains subject to
polynomial-approximation lower bounds.

Under the standard block encoding of $H/\norm{H}$, the direct even QSVT
polynomial approximates
$\mathrm e^{-T\norm{H}^{\alpha}\abs{x}^{\alpha}}$ on $[-1,1]$.  It
is entire when $\alpha$ is an even positive integer; odd
integer and noninteger powers retain a singularity at the interior
point $x=0$.  Under the shifted signal, the scalar target becomes
$\mathrm e^{-T\norm{H}^{\alpha}[(1+x)/2]^\alpha}$.  The exact change
of variables $x=2y^2-1$ identifies its degree-$d$ approximation with
an even degree-$2d$ approximation to
$\mathrm e^{-T\norm{H}^{\alpha}\abs{y}^{2\alpha}}$.  All positive
integers then enter the entire-function regime, while a
noninteger endpoint singularity has twice the approximation exponent
of the corresponding interior singularity.  This quadratic lifting is
the source of both the improved physical scale and the improved
precision dependence proved below.

A broader Poisson-summation framework for quantum matrix
transformations was introduced previously~\cite{wang2026unified}.  Here
we focus on power-exponential dissipation and pursue a different
implementation: the contour-based route is removed, the coherent
frequency LCU is compiled into a polynomial transformation, and direct and
shifted signal access are treated separately.  This leads to the
access-dependent approximation regimes and end-to-end query bounds
established below.

The polynomial construction applies directly to heat, biharmonic, and
fractional dissipation.  The two non-Hermitian applications use distinct
structures.  For time-independent non-Hermitian simulation, we combine
Poisson summation with Weyl calculus
~\cite{ni2026quantumeigenvaluetransformationlinear,
weyl1950theory,anderson1969weyl,jefferies2004spectral}, the optimal LCHS
kernel~\cite{PhysRevLett.131.150603,
an2023quantumalgorithmlinearnonunitary,
low2025optimalquantumsimulationlinear}, and existing sinh--sinh
quadrature~\cite{aftab2026lchsmpf}.  The entire scalar surrogate has a
Fourier representation that Weyl calculus lifts to the noncommuting
Hermitian pair.  Strip analyticity and Gaussian decay of its Fourier
kernel control contour displacement and alias decay, whereas the exact
Poisson identity follows from the Fourier representation; it is not a
substitution rule for an ordinary holomorphic functional calculus of
$L+\mathrm iG$.  The resulting alias
identity does not alter the optimal LCHS query order or LCU
normalization.  The sinh--sinh map can reduce a precision-dominated node
count, but the uniform count retains a linear $T\beta_A$ factor.
Amplitude--phase separation (APS) instead exploits the semigroup law for
$\mathrm e^{-tH}$.  Its dissipative Dyson products divide the total time
into nonnegative intervals, so one controlled semigroup family supplies
all dissipative factors.  For $\exp[-T(H+\mathrm iG)]$, we use half-root
access to
$(H/\beta_H)^{\frac{1}{2\alpha}}$ for fixed positive integer $\alpha$
~\cite{hu2026quantumsimulationnonunitarydynamics}.  A degree-two QSVT
sequence also constructs the associated root-shift signal
$2(H/\beta_H)^{\frac{1}{\alpha}}-I$ exactly
~\cite{Gilyen2019qsvt}.  For APS we give the QSVT degree required for
one controlled dissipative-family call, the number of such calls made by
the interaction-picture algorithm, and the primitive half-root query
count obtained by implementing that family with its QSVT circuit.
Both non-Hermitian applications are
stated under their additional oracle assumptions.

\section{Poisson-Compiled QSVT for Power-Exponential Functions}
\label{sec:poisson_qsvt}

Here and in the appendices, $C$ denotes a positive
universal constant and $C_\alpha,c_\alpha$ denote positive constants
depending only on $\alpha$.  Their values may change from line to line.
Subscripts on asymptotic symbols are omitted.  Their hidden
constants may depend on parameters explicitly held fixed in the
corresponding limit, as stated in the surrounding text.
We omit additive constant query and sample costs and understand each
resource count to be at least one.  If the target
differs from the identity by no more than the requested error, no
transformation is required.  This convention avoids adding $1+$ to
every complexity formula.
Logarithmic factors are written in the high-precision convention
$1/\epsilon\gg1$, or $u_r/\epsilon\gg1$ for state preparation.  A
logarithm or iterated logarithm appearing in a denominator is
understood to be lower-bounded by one; bounded lower-precision regimes
are absorbed into the implicit constants.

\subsection{Problem formulation and access models}
\label{subsec:problem}

Let $H\ne0$ be Hermitian and positive semidefinite.  We assume access
to a unitary $U_H$ whose designated block is exactly, in the standard
block-encoding sense~\cite{Gilyen2019qsvt},
\begin{equation}
    (\bra{0^a}\otimes I)U_H(\ket{0^a}\otimes I)
    =
    \frac{H}{\norm{H}}.
    \label{eq:block_encoding}
\end{equation}
The second access model supplies a unitary $U_S$ whose designated
block is
\begin{equation}
    (\bra{0^b}\otimes I)U_S(\ket{0^b}\otimes I)
    =
    S,
    \qquad
    S:=\frac{2H}{\norm{H}}-I.
    \label{eq:shifted_block_encoding}
\end{equation}
Eq.~\eqref{eq:shifted_block_encoding} is a unit-normalized block
encoding and is an independent access assumption.  A generic block
encoding in Eq.~\eqref{eq:block_encoding} does not automatically give
Eq.~\eqref{eq:shifted_block_encoding} with unit normalization: an
LCU construction of $2H/\norm{H}-I$ would have coefficient one-norm
three.  This mathematical distinction need not imply a substantial
implementation difference in the sparse-matrix model.  If $H$ is
$s$-sparse and its entries and nonzero locations are computed by the
usual sparse-access oracles, the same location oracle and an entry
oracle that rescales by two and subtracts the diagonal give $S$ with
sparsity at most $s+1$ and constant query
overhead~\cite{Berry2007,Gilyen2019qsvt}.  Unit
normalization must still be verified for the particular sparse
encoding.  It holds for the finite-difference Laplacians in
Sec.~\ref{subsec:pde}, where each row of the shifted signal has
absolute sum one.

For every $\alpha>0$, $H^\alpha$ is defined by the spectral calculus;
zero eigenvalues cause no singularity in the target exponential.
The zero matrix is trivial and will be excluded.  Since
$0\preceq H/\norm{H}\preceq I$, the scalar transformations in the two
access models are, respectively,
\begin{equation}
\begin{aligned}
    x&\longmapsto
    \mathrm e^{-T\norm{H}^{\alpha}x^\alpha},
    &&0\leq x\leq1,\\
    s&\longmapsto
    \mathrm e^{-T\norm{H}^{\alpha}[(1+s)/2]^\alpha},
    &&-1\leq s\leq1.
\end{aligned}
    \label{eq:normalized_target}
\end{equation}
If only a block encoding with a known normalization larger than
$\norm{H}$ is available, every statement remains valid after replacing
$\norm{H}$ by that normalization in both the oracle and the formulas.
In particular, the shifted block is then $2H/\beta_H-I$ for a known
$\beta_H\geq\norm{H}$.  The present normalization is used to keep the
formulas sharp.
The contraction bounds
$\mathrm e^{-T\norm{H}^{\alpha}}I\preceq
\mathrm e^{-T H^\alpha}\preceq I$ imply
$1\leq u_r\leq\mathrm e^{T\norm{H}^{\alpha}}$ for every nonzero
input.  No spectral gap is assumed.

QSVT requires a polynomial with a parity fixed by its degree.  We use
an even degree and therefore approximate
$\mathrm e^{-T\norm{H}^{\alpha}\abs{x}^{\alpha}}$ on $[-1,1]$.
On the nonnegative spectrum of $H/\norm{H}$ this is exactly
Eq.~\eqref{eq:normalized_target}.  If a real even polynomial $P_d$
obeys
\begin{equation}
    \max_{x\in[-1,1]}
    \abs{
    P_d(x)-
    \mathrm e^{-T\norm{H}^{\alpha}\abs{x}^{\alpha}}
    }
    \leq\epsilon,
    \label{eq:uniform_scalar_error}
\end{equation}
then the spectral theorem gives
\begin{equation}
    \norm{
    P_d\!\left(H/\norm{H}\right)
    -\mathrm e^{-T H^\alpha}
    }
    \leq\epsilon.
    \label{eq:operator_from_scalar}
\end{equation}
This parity choice is forced in the gap-independent, single-sequence
QSVT model.  An odd QSVT polynomial vanishes at the origin, whereas the
target equals one there, and hence cannot achieve error smaller than
one uniformly over positive semidefinite inputs that may have a zero
eigenvalue.  Thus every nontrivial ordinary-QSVT approximation used
under standard access is even.  Indefinite-parity generalized
eigenvalue transformations form a different implementation class.
The QSVT completion theorem states that a real polynomial of parity
$d\bmod 2$ and modulus at most one on $[-1,1]$ is implementable, up to
a constant block normalization, by $d$ uses of the signal
unitary~\cite{Gilyen2019qsvt}.  A polynomial in the shifted variable
generally has mixed parity.  Generalized Hermitian eigenvalue
transformation implements it with the same query degree, but its
worst-case block normalization can grow as $\Ocal(\log(d+2))$
~\cite{sunderhauf2023generalized}.  For the query bounds below we
instead combine the even and odd ordinary-QSVT components with one
additional ancilla.  The fixed factor $1/2$ in
Eq.~\eqref{eq:qsvt_rescaling} keeps both components admissible and gives
a degree-independent constant normalization.  This constant-size parity
handling is not a frequency LCU and does not depend on the number of
Poisson samples.

\subsection{Poisson-generated polynomial construction}
\label{subsec:poisson_polynomial}

We use the Fourier-transform convention
in~\cite{stein1971introduction,katznelson2004introduction},
\begin{equation}
    \widehat f(\xi)
    =
    \int_{\mathbb R}f(y)\,
    \mathrm e^{-2\pi\mathrm i y\xi}\,\mathrm dy,
    \qquad
    f(y)
    =
    \int_{\mathbb R}\widehat f(\xi)\,
    \mathrm e^{2\pi\mathrm i y\xi}\,\mathrm d\xi.
    \label{eq:fourier_convention}
\end{equation}
Only one kernel is needed:
\begin{equation}
    f_{\alpha,T}(y)
    :=
    \int_{\mathbb R}
    \mathrm e^{-T\norm{H}^{\alpha}\abs{\xi}^{\alpha}}
    \mathrm e^{2\pi\mathrm i y\xi}\,\mathrm d\xi.
    \label{eq:inverse_fourier_kernel}
\end{equation}
It is real and even.  The periodization proof in
Appendix~\ref{app:kernel} shows that both sides below converge
absolutely and locally uniformly.  Poisson summation with sampling
interval $h>0$ therefore gives the exact identity
\begin{equation}
    h\sum_{k\in\mathbb Z}
    f_{\alpha,T}(kh)
    \mathrm e^{-2\pi\mathrm i khx}
    =
    \sum_{n\in\mathbb Z}
    \mathrm e^{-T\norm{H}^{\alpha}
    \abs{x+n/h}^{\alpha}}.
    \label{eq:poisson_exact}
\end{equation}
The factors $2\pi$ in Eqs.~\eqref{eq:fourier_convention} and
\eqref{eq:poisson_exact} fix the Fourier convention; using angular
frequency changes both the oscillatory phase and the alias spacing.

Truncating the sampled kernel at $\abs{k}\leq N$ and isolating the
$n=0$ term yields
\begin{align}
    &h\sum_{k=-N}^{N}
    f_{\alpha,T}(kh)
    \mathrm e^{-2\pi\mathrm i khx}
    -
    \mathrm e^{-T\norm{H}^{\alpha}\abs{x}^{\alpha}}
    \nonumber\\
    &\quad=
    \sum_{n\ne0}
    \mathrm e^{-T\norm{H}^{\alpha}
    \abs{x+n/h}^{\alpha}}
    -
    h\sum_{\abs{k}>N}
    f_{\alpha,T}(kh)
    \mathrm e^{-2\pi\mathrm i khx}.
    \label{eq:poisson_alias_tail}
\end{align}
Eq.~\eqref{eq:poisson_alias_tail} is the starting point of the
error analysis.  Its two terms are the spectral-aliasing error and the
sampled-kernel tail.

Because $f_{\alpha,T}$ is even, the finite sum is real:
\begin{equation}
    h f_{\alpha,T}(0)
    +
    2h\sum_{k=1}^{N}
    f_{\alpha,T}(kh)
    \cos(2\pi khx).
    \label{eq:cosine_poisson_sum}
\end{equation}
The Jacobi--Anger expansion~\cite{NIST:DLMF}
\begin{equation}
    \cos(zx)
    =
    J_0(z)
    +
    2\sum_{j=1}^{\infty}
    (-1)^jJ_{2j}(z)T_{2j}(x)
    \label{eq:jacobi_anger_cosine}
\end{equation}
therefore gives the explicit even polynomial
\begin{align}
    P_d(x)
    &:={}
    h f_{\alpha,T}(0)
    +2h\sum_{k=1}^{N}f_{\alpha,T}(kh)
    \Bigg[
        J_0(2\pi kh)
        \nonumber\\[-1mm]
    &\hspace{18mm}
        +2\sum_{j=1}^{\frac{d}{2}}
        (-1)^jJ_{2j}(2\pi kh)
        T_{2j}(x)
    \Bigg],
    \label{eq:poisson_compiled_polynomial}
\end{align}
where $d$ is even.  All sums over $k$ are evaluated classically and
compiled into one list of QSVT phases.  Eq.~\eqref{eq:poisson_compiled_polynomial} is not a coherent LCU of the
$2N+1$ Fourier modes.

The polynomial in Eq.~\eqref{eq:poisson_compiled_polynomial} is the
degree-$d$ Chebyshev projection of
Eq.~\eqref{eq:cosine_poisson_sum}.  Since the projection norm is
$\Ocal(\log(d+2))$~\cite{mason2003chebyshev,
trefethen2019approximation}, it is sufficient to make the right-hand side of
Eq.~\eqref{eq:poisson_alias_tail} no larger than
$\epsilon/[C\log(d+2)]$.  The proof in
Appendix~\ref{app:poisson_compilation} gives explicit choices of $h$
and $N$.  This projection argument aggregates cancellations among all
sampled frequencies and avoids charging the QSVT degree for
approximating every cosine separately.

\begin{theorem}[Poisson compilation]
\label{thm:poisson_compilation}
Let $\alpha,T>0$ and $0<\epsilon<1/4$.  There are an even integer $d$,
a spacing $h$, and a cutoff $N$ such that the polynomial $P_d$ in
Eq.~\eqref{eq:poisson_compiled_polynomial} satisfies
\begin{equation}
    \max_{x\in[-1,1]}
    \abs{
    P_d(x)-
    \mathrm e^{-T\norm{H}^{\alpha}\abs{x}^{\alpha}}
    }
    \leq\frac{\epsilon}{2}.
    \label{eq:poisson_polynomial_error}
\end{equation}
After computing the coefficients to the accuracy stated in
Appendix~\ref{app:poisson_compilation}, the fixed rescaling
\begin{equation}
    \frac{P_d(x)}{2(1+\epsilon/2)}
    \label{eq:qsvt_rescaling}
\end{equation}
is QSVT admissible.  If $B$ denotes its designated QSVT block, then
\begin{equation}
    \norm{2B-\mathrm e^{-T H^\alpha}}
    \leq\epsilon.
    \label{eq:qsvt_normalization_two_error}
\end{equation}
Thus $B$ is a normalization-two block encoding of
$\mathrm e^{-T H^\alpha}$ with operator error at most $\epsilon$.
\end{theorem}

\begin{proposition}[Quadratic lifting for the shifted signal]
\label{prop:quadratic_lifting}
Let $P_{2d}$ be the even Poisson-compiled polynomial for the scalar
target
$\mathrm e^{-T\norm{H}^{\alpha}\abs{y}^{2\alpha}}$; equivalently,
apply Theorem~\ref{thm:poisson_compilation} with exponent $2\alpha$
and replace $\norm{H}$ there by $\sqrt{\norm{H}}$.  Then
\begin{equation}
    P_{2d}\!\left(\sqrt{\frac{1+s}{2}}\right)
    \label{eq:shifted_lifted_polynomial}
\end{equation}
is a polynomial of degree at most $d$ in $s$ and satisfies
\begin{equation}
\max_{s\in[-1,1]}
\left|
P_{2d}\!\left(\sqrt{\frac{1+s}{2}}\right)
-
\mathrm e^{-T\norm{H}^{\alpha}[(1+s)/2]^\alpha}
\right|
\leq\frac{\epsilon}{2}.
\label{eq:shifted_lifted_error}
\end{equation}
    After the same boundedness rescaling as in
    Eq.~\eqref{eq:qsvt_rescaling}, decompose this mixed-parity polynomial
    into its even and odd parts.  Coherently selecting the two
    ordinary-QSVT circuits gives a constant-normalization block encoding
    of $\mathrm e^{-T H^\alpha}$ using $\Ocal(d)$ calls to $U_S$ and
    $U_S^\dagger$.
\end{proposition}

\begin{proof}
Write $P_{2d}(y)=\sum_{j=0}^{d}c_jy^{2j}$.  The expression in
Eq.~\eqref{eq:shifted_lifted_polynomial} is exactly
$\sum_{j=0}^{d}c_j[(1+s)/2]^j$, so the displayed square root introduces
no branch and the result is a polynomial of degree at most $d$.
For this substitution,
\begin{equation}
\mathrm e^{-T(\sqrt{\norm{H}})^{2\alpha}
\abs{y}^{2\alpha}}
=
\mathrm e^{-T\norm{H}^{\alpha}[(1+s)/2]^\alpha}
\end{equation}
when $y=\sqrt{(1+s)/2}$.  The uniform error and boundedness therefore
follow directly from Theorem~\ref{thm:poisson_compilation}.  Finally,
$(I+S)/2=H/\norm{H}$, so the spectral theorem gives the asserted
operator transformation.
\end{proof}

The proof, including finite-precision evaluation of the kernel
samples, is given in Appendix~\ref{app:poisson_compilation}.
The factor $1/2$ in Eq.~\eqref{eq:qsvt_rescaling} is a universal
admissibility margin.  It only changes the state-preparation
probability by a constant and is unrelated to a Fourier LCU norm.
The $2N+1$ kernel values and the Bessel sums in
Eq.~\eqref{eq:poisson_compiled_polynomial} are classical data.  They
require $\Ocal(Nd)$ arithmetic operations once the kernel values are
available.  Standard QSP factorization then converts the resulting
$d+1$ coefficients into a phase sequence~\cite{Gilyen2019qsvt,
haah2019product}.  These costs do not change the number of calls to
$U_H$ or $U_S$, but they are part of the classical preprocessing
required to instantiate the circuit.
For general $\alpha$, the bounds below specify the number and precision
of the required one-dimensional kernel evaluations.  They do not assert
an optimized bit complexity for a particular numerical quadrature
implementation.

\subsection{Integer and noninteger powers}
\label{subsec:powers}

The inverse Fourier kernel determines the classical sample count,
while the regularity of
$\mathrm e^{-T\norm{H}^{\alpha}\abs{x}^{\alpha}}$ determines the QSVT
degree.  These two resources should not be conflated.

If $\alpha$ is an even positive integer, complex translation of the
integral in Eq.~\eqref{eq:inverse_fourier_kernel} gives the
steepest-descent bound
\begin{equation}
    \abs{f_{\alpha,T}(y)}
    \leq
    \frac{C_\alpha}{T^{\frac{1}{\alpha}}\norm{H}}
    \exp\!\left[
        -c_\alpha
        \left(
        \frac{\abs{y}}{T^{\frac{1}{\alpha}}\norm{H}}
        \right)^{\frac{\alpha}{\alpha-1}}
    \right].
    \label{eq:even_kernel_decay}
\end{equation}
Thus the kernel has a stretched-exponential tail.  If instead
$\alpha$ is not an even integer, the nonanalyticity at the origin and
the Fourier form of Watson's lemma give~\cite{wong2001asymptotic}
\begin{equation}
    f_{\alpha,T}(y)
    =
    \frac{
    2T\norm{H}^{\alpha}
    \Gamma(\alpha+1)\sin(\pi\alpha/2)}
    {(2\pi\abs{y})^{\alpha+1}}
    +
    o\!\left(\abs{y}^{-\alpha-1}\right).
    \label{eq:noneven_kernel_asymptotic}
\end{equation}
The coefficient in Eq.~\eqref{eq:noneven_kernel_asymptotic} is
nonzero exactly when $\alpha$ is not an even integer.  Odd integers
therefore have the same algebraic kernel mechanism as fractional
powers.  For $\alpha>2$ the kernel need not be nonnegative; this does
not obstruct the construction because its samples are aggregated
classically rather than prepared as LCU amplitudes.  Explicit choices
of $h$ and $N$ are deferred to
Appendix~\ref{app:kernel}, because they affect only classical
coefficient generation and not the number of oracle calls.

For the shifted signal, Proposition~\ref{prop:quadratic_lifting}
repeats the same analysis with exponent $2\alpha$.  This exponent is
an even integer exactly when $\alpha$ is a positive integer.  Hence
odd integer powers acquire the same stretched-exponential kernel as
even integer powers.  If $\alpha$ is noninteger, the lifted kernel has
tail $\Ocal(\abs{y}^{-2\alpha-1})$, improving the truncation
power from $\alpha$ to $2\alpha$.

\begin{theorem}[Polynomial degree in the two access models]
\label{thm:degree}
Let $0<\epsilon<1/4$.  Under access to $U_H$, if $\alpha$ is an even
positive integer, the Poisson-compiled polynomial can be chosen with
\begin{equation}
    d
    =
    \Ocal\!\left(
        \left(
        T\norm{H}^{\alpha}+\log\frac{1}{\epsilon}
        \right)^{\frac{1}{\alpha}}
        \log^{1-\frac{1}{\alpha}}\frac{1}{\epsilon}
    \right).
    \label{eq:even_degree}
\end{equation}
If $\alpha$ is an odd positive integer or is noninteger, the
parity-compatible construction under $U_H$ can be chosen with
\begin{equation}
    d
    =
    \Ocal\!\left(
        \norm{H}
        \left(\frac{T}{\epsilon}\right)^{\frac{1}{\alpha}}
    \right).
    \label{eq:noneven_degree}
\end{equation}
Under access to $U_S$, every positive integer $\alpha$ admits degree
\begin{equation}
    d
    =
    \Ocal\!\left(
        \left(
        T\norm{H}^{\alpha}+\log\frac{1}{\epsilon}
        \right)^{\frac{1}{2\alpha}}
        \log^{1-\frac{1}{2\alpha}}\frac{1}{\epsilon}
    \right).
    \label{eq:shifted_integer_degree}
\end{equation}
If $\alpha$ is noninteger, it is sufficient under $U_S$ to use
\begin{equation}
    d
    =
    \Ocal\!\left(
        \sqrt{\norm{H}}
        \left(\frac{T}{\epsilon}\right)^{\frac{1}{2\alpha}}
    \right).
    \label{eq:shifted_fractional_degree}
\end{equation}
When $T\norm{H}^{\alpha}\leq\epsilon$, the identity transformation already has the
required accuracy and the nontrivial terms in these bounds may be
omitted.
\end{theorem}

The single-factor forms in Eqs.~\eqref{eq:even_degree} and
\eqref{eq:shifted_integer_degree} are uniform upper bounds.  Up to
constants depending only on $\alpha$, they are equivalent to writing
the physical-scale contribution and the high-precision logarithm as
two additive terms.  The sharper joint interpolation is given in
Eqs.~\eqref{eq:even_degree_interpolation} and
\eqref{eq:shifted_integer_degree_interpolation} of
Appendix~\ref{app:poisson_compilation}.  In particular, when
$T\norm{H}^{\alpha}>0$ is fixed and $\epsilon\to0$, the uniform
$\log(1/\epsilon)$ behavior can be sharpened to
$\log(1/\epsilon)/\log\log(1/\epsilon)$.  This fixed-scale refinement
is the one used in the precision column of
Table~\ref{tab:degree_regimes}.

For $\alpha=1$, Eq.~\eqref{eq:shifted_integer_degree} becomes
\begin{equation}
    d=\Ocal\!\left(
    \sqrt{
    \left(T\norm{H}+\log\frac{1}{\epsilon}\right)
    \log\frac{1}{\epsilon}}
    \right).
    \label{eq:shifted_square_root_uniform}
\end{equation}
The bound in Eq.~\eqref{eq:shifted_square_root_uniform} is
$\Ocal(\sqrt{T\norm{H}\log(1/\epsilon)})$ whenever
$T\norm{H}\geq\log(1/\epsilon)$.  In this fast-forwarding regime,
direct access to $2H/\norm{H}-I$ reproduces the square-root dependence
used by amplitude--phase separation
~\cite{aggarwal2022optimaldegree,
low2017hamiltoniansimulationuniformspectral,
hu2026quantumsimulationnonunitarydynamics}.  Outside this regime,
Eq.~\eqref{eq:shifted_square_root_uniform} retains the uniform
high-precision correction; Eq.~\eqref{eq:shifted_integer_degree_interpolation}
is sharper still.

For even integer $\alpha$ under standard access, the physical-scale part of
Eq.~\eqref{eq:even_degree} is the moderate-precision contribution from
the stretched-exponential Fourier tail.  The uniform high-precision
part is a simplification of the entire-function contribution in
Eq.~\eqref{eq:even_degree_interpolation}.  For non-even $\alpha$,
Eq.~\eqref{eq:noneven_degree} is set by the interior nonanalyticity.
Under shifted access, Eq.~\eqref{eq:shifted_integer_degree} follows by
making the replacements $\alpha\mapsto2\alpha$ and
$\norm{H}\mapsto\sqrt{\norm{H}}$ before the quadratic lifting.  The
same replacement in the algebraic case gives
Eq.~\eqref{eq:shifted_fractional_degree}.  Thus odd integers move into
the entire-function regime, whereas noninteger powers remain
algebraic but gain a factor of two in the precision exponent.
The direct estimates used here are classical Chebyshev and
modulus-of-smoothness bounds~\cite{devore1993constructive,
ditzian1987moduli,trefethen2019approximation}.  Their converse forms
will be used below only in the parameter regimes for which they are
uniform.

\begin{table*}[t]
\caption{Approximation regimes of the Poisson constructions in the two
access models.  The standard-oracle non-even precision entry is a
matching upper and lower bound for gap-independent ordinary QSVT with
a single phase sequence.  It is not a minimax claim for generalized
indefinite-parity transformations.  The
physical-scale column is asymptotic as
$T\norm{H}^{\alpha}\to\infty$ at fixed error.  The precision column
keeps $T\norm{H}^{\alpha}>0$ fixed.  Each $\Theta$ entry is a
one-parameter statement in the indicated limit; the table does not
claim that the uniform joint upper bounds are minimax for all
time--precision scalings.}
\label{tab:degree_regimes}
\begin{ruledtabular}
\begin{tabular}{llcc}
access & $\alpha$ & fixed-error degree & fixed-scale precision degree \\
\hline
 $H/\norm{H}$ & even positive integer
& $\Theta((T\norm{H}^{\alpha})^{\frac{1}{\alpha}})$
& $\Theta\!\left(\frac{\log(1/\epsilon)}{\log\log(1/\epsilon)}\right)$ \\
 $H/\norm{H}$ & odd positive integer or noninteger
& $\Theta((T\norm{H}^{\alpha})^{\frac{1}{\alpha}})$
& $\Theta(\epsilon^{-\frac{1}{\alpha}})$
\\
 $2H/\norm{H}-I$ & positive integer
& $\Theta((T\norm{H}^{\alpha})^{\frac{1}{2\alpha}})$
& $\Theta\!\left(\frac{\log(1/\epsilon)}{\log\log(1/\epsilon)}\right)$ \\
 $2H/\norm{H}-I$ & noninteger
& $\Theta((T\norm{H}^{\alpha})^{\frac{1}{2\alpha}})$
& $\Theta(\epsilon^{-\frac{1}{2\alpha}})$
\end{tabular}
\end{ruledtabular}
\end{table*}

\subsection{End-to-end complexity and approximation limits}
\label{subsec:complexity}

One application of a degree-$d$ polynomial eigenvalue transformation
uses $\Ocal(d)$ calls to the corresponding signal oracle and its
inverse.  Applied to the
normalized input $\ket{u_0}=u_0/\norm{u_0}$, the desired output block
in Theorem~\ref{thm:poisson_compilation} has norm
$\Theta(1/u_r)$.  Fixed-point amplitude amplification therefore
prepares a state proportional to $u_T$ with
$\Ocal(u_r d)$ block-encoding queries, assuming an upper bound on
$u_r$~\cite{Yoder2014fixedpoint}.  The stated cost gives constant
success probability; standard repetition gives any prescribed higher
success probability.

\begin{theorem}[End-to-end state preparation]
\label{thm:end_to_end}
Let $0<\epsilon<1/2$, and require Euclidean error at most $\epsilon$
in the normalized output state.  For an even positive integer
$\alpha$, the number of calls to $U_H$ and $U_H^\dagger$ is
\begin{equation}
    \Ocal\!\left(
    u_r
        \left(
        T\norm{H}^{\alpha}+\log\frac{u_r}{\epsilon}
        \right)^{\frac{1}{\alpha}}
        \log^{1-\frac{1}{\alpha}}
        \frac{u_r}{\epsilon}
    \right).
    \label{eq:even_end_to_end}
\end{equation}
If $\alpha$ is an odd positive integer or is noninteger, the
corresponding number is
\begin{equation}
    \Ocal\!\left(
    u_r\norm{H}
    \left(\frac{T u_r}{\epsilon}\right)^{\frac{1}{\alpha}}
    \right).
    \label{eq:noneven_end_to_end}
\end{equation}
Under access to $U_S$, every positive integer $\alpha$ instead uses
\begin{equation}
    \Ocal\!\left(
    u_r
        \left(
        T\norm{H}^{\alpha}+\log\frac{u_r}{\epsilon}
        \right)^{\frac{1}{2\alpha}}
        \log^{1-\frac{1}{2\alpha}}
        \frac{u_r}{\epsilon}
    \right)
    \label{eq:shifted_integer_end_to_end}
\end{equation}
calls to $U_S$ and $U_S^\dagger$.  For noninteger $\alpha$, a
sufficient shifted-oracle query bound is
\begin{equation}
    \Ocal\!\left(
    u_r\sqrt{\norm{H}}
    \left(\frac{T u_r}{\epsilon}\right)^{\frac{1}{2\alpha}}
    \right).
    \label{eq:shifted_fractional_end_to_end}
\end{equation}
\end{theorem}

For the integer cases, fixing $T\norm{H}^{\alpha}$ and $u_r$ sharpens
the uniform high-precision order implicit in
Eqs.~\eqref{eq:even_end_to_end} and
\eqref{eq:shifted_integer_end_to_end} from $\log(u_r/\epsilon)$ to
$\log(u_r/\epsilon)/\log\log(u_r/\epsilon)$.  The longer interpolation
formula is needed only when time, normalization, and precision vary
simultaneously.

The factor $u_r$ has two roles.  To obtain normalized-state error
$\epsilon$, Theorem~\ref{thm:degree} is invoked with polynomial error
of order $\epsilon/u_r$, and amplitude amplification requires
$\Theta(u_r)$ uses of that polynomial.  There is no additional
Fourier coefficient one-norm in
Eqs.~\eqref{eq:even_end_to_end} and
\eqref{eq:noneven_end_to_end}, nor in their shifted counterparts.

The query bounds treat $U_H$ or $U_S$ as an exact oracle.  If each
implemented call differs from the corresponding ideal oracle by at
most $\epsilon_H$ in operator norm, a
telescoping argument bounds the error of one degree-$d$ polynomial
transformation
by $d\epsilon_H$.  It is therefore sufficient to synthesize each
oracle call and each elementary QSP rotation to error
$\Ocal(\epsilon/(u_r d))$.  This requirement affects gate synthesis,
not the number of block-encoding queries.

The degree estimates also have a direct approximation-theoretic
interpretation.  For fixed even integer $\alpha$, classical Bernstein
theory for entire functions of finite order explains the two regimes
in Eq.~\eqref{eq:even_degree_interpolation}.  At constant error, an interior
Bernstein inequality implies the matching lower bound
$\Omega(\norm{H}T^{\frac{1}{\alpha}})$.  At fixed $T$ and $H$, as
$\epsilon\rightarrow0$, finite-order entire-function theory gives the
matching order
$\Theta(\log(1/\epsilon)/\log\log(1/\epsilon))$~\cite{bernstein1912ordre,
timan1963theory,trefethen2019approximation}.
For $\alpha=2$, the
full joint time--precision expression agrees with the known optimal
degree for Gaussian exponentials~\cite{aggarwal2022optimaldegree}.

If $\alpha$ is not an even integer, the best approximation error for
the even extension implemented in this paper is proportional to
$d^{-\alpha}$ as $d\to\infty$, with constants depending on $\alpha$,
$T$, and $H$~\cite{bernstein1914sur,devore1993constructive,
ditzian1987moduli}.  Hence the
$\epsilon^{-\frac{1}{\alpha}}$ exponent in
Eq.~\eqref{eq:noneven_degree} is optimal for gap-independent ordinary
QSVT with one phase sequence: an odd sequence already has unit error at
the origin.  Together with the uniform
constant-error lower bound below, this gives the two statements
recorded in Table~\ref{tab:degree_regimes}.  We do not claim a matching
two-parameter lower bound with constants uniform in both
$T\norm{H}^{\alpha}$ and $\epsilon$ for every non-even power, or a
lower bound against generalized indefinite-parity constructions.

For $U_S$, Proposition~\ref{prop:quadratic_lifting} identifies
degree-$d$ polynomials in the shifted variable with even degree-$2d$
polynomials in the lifted variable. Consequently, integer $\alpha$
yields entire-function approximation of order $2\alpha$, whereas
noninteger $\alpha$ has fixed-scale error $\Theta(d^{-2\alpha})$.
Markov's endpoint inequality further gives the matching fixed-error
lower bound
\[
\Omega\!\left(
(T\norm{H}^{\alpha})^{\frac{1}{2\alpha}}
\right).
\]
Thus all entries in Table~\ref{tab:degree_regimes} are tight in the two
separate limits stated in its caption. The standard-signal non-even
result is optimal only for ordinary single-sequence QSVT, and no
general joint two-parameter minimax claim is made.

\section{Applications}
\label{sec:applications}

\subsection{Heat, biharmonic, and fractional dissipation}
\label{subsec:pde}

Let $\mathcal L\succeq0$ be a spatially discretized elliptic operator.
The semidiscrete dissipative equation
\begin{equation}
    \frac{\mathrm d u(t)}{\mathrm dt}
    =
    -\kappa\mathcal L^\alpha u(t)
    \label{eq:pde_family}
\end{equation}
has solution
\begin{equation}
    u(T)=
    \mathrm e^{-\kappa T\mathcal L^\alpha}u(0).
    \label{eq:pde_solution}
\end{equation}
Quantum partial differential equation (PDE) algorithms based on spectral discretization, finite
differences, time marching, Dyson series, and Schr\"odingerization
provide complementary access models and error accounts
~\cite{childs2020spectral,JIN2022111641,PhysRevA.108.032603,
Fang2023timemarchingbased,Berry_2024,PhysRevLett.133.230602,
jin2025schrodingerizationmethodlinearnonunitary,Gonzalez_Conde_2023}.
Here the spatial
discretization is fixed, and the task is the matrix-function step in
Eq.~\eqref{eq:pde_solution}.

Let $\beta_{\mathcal L}\geq\norm{\mathcal L}$ be the known block
normalization.  We distinguish access to
$\mathcal L/\beta_{\mathcal L}$ from direct unit-normalized access to
$2\mathcal L/\beta_{\mathcal L}-I$.  The following costs prepare the
normalized state proportional to $u(T)$ with Euclidean error at most
$\epsilon$ and constant success probability.  The corresponding
operator-block costs follow by removing the outer factor $u_r$ and
replacing every occurrence of $u_r$ inside the formulas by one.

\begin{samepage}
For the heat equation, $\mathcal L$ discretizes $-\Delta$ and
$\alpha=1$.  The resulting query bounds are displayed below.
\begin{equation}
\begin{aligned}
\text{standard:}\quad
&\Ocal\!\left(
\frac{u_r^2\beta_{\mathcal L}\kappa T}{\epsilon}
\right),
\\[2mm]
\text{shifted:}\quad
&\Ocal\!\left(
u_r\left(
\beta_{\mathcal L}\kappa T
+\log\frac{u_r}{\epsilon}
\right)^{\frac{1}{2}}
\log^{\frac{1}{2}}\!\frac{u_r}{\epsilon}
\right).
\end{aligned}
\label{eq:heat_application_queries}
\end{equation}
\end{samepage}
Thus direct shifted access changes both the physical dependence and the
precision dependence of the standard even extension.  When
$\beta_{\mathcal L}\kappa T\geq\log(u_r/\epsilon)$, the shifted bound is
$\Ocal(u_r\sqrt{\beta_{\mathcal L}\kappa T\log(u_r/\epsilon)})$.
This is the square-root dissipative dependence used by APS, now obtained
directly from the shifted signal by one constant-normalization polynomial
eigenvalue transformation
~\cite{aggarwal2022optimaldegree,
low2017hamiltoniansimulationuniformspectral,
hu2026quantumsimulationnonunitarydynamics}.

For biharmonic dissipation,
$\partial_tu=-\kappa(-\Delta)^2u$, one sets $\alpha=2$.  This model
appears, for example, in linearized surface-diffusion and thermal
grooving problems~\cite{mullins1957thermal}.  The explicit query bounds
are
\begin{equation}
\begin{aligned}
\text{standard:}\quad
&\Ocal\!\left(
u_r\left(
\beta_{\mathcal L}^{2}\kappa T
+\log\frac{u_r}{\epsilon}
\right)^{\frac{1}{2}}
\log^{\frac{1}{2}}\!\frac{u_r}{\epsilon}
\right),
\\[2mm]
\text{shifted:}\quad
&\Ocal\!\left(
u_r\left(
\beta_{\mathcal L}^{2}\kappa T
+\log\frac{u_r}{\epsilon}
\right)^{\frac{1}{4}}
\log^{\frac{3}{4}}\!\frac{u_r}{\epsilon}
\right).
\end{aligned}
\label{eq:biharmonic_application_queries}
\end{equation}
The standard formula is the Gaussian regime whose joint
time--precision degree is minimax optimal
~\cite{aggarwal2022optimaldegree}.  The shifted formula is conditional
on the stronger signal access and has the fourth-root scale
$(\kappa T\beta_{\mathcal L}^{2})^{\frac{1}{4}}$.  In both integer
examples, fixing $\kappa T\beta_{\mathcal L}^{\alpha}$ sharpens the
uniform high-precision order to
$\log(u_r/\epsilon)/\log\log(u_r/\epsilon)$.

For every noninteger $\alpha>0$, fractional diffusion gives
\begin{equation}
\begin{aligned}
\text{standard:}\quad
&\Ocal\!\left(
u_r\beta_{\mathcal L}
\left(\frac{\kappa T u_r}{\epsilon}\right)^{\frac{1}{\alpha}}
\right),
\\[1mm]
\text{shifted:}\quad
&\Ocal\!\left(
u_r\sqrt{\beta_{\mathcal L}}
\left(\frac{\kappa T u_r}{\epsilon}\right)^{\frac{1}{2\alpha}}
\right).
\end{aligned}
\label{eq:fractional_application_queries}
\end{equation}
No oracle for $\mathcal L^\alpha$ is assumed.  The fractional-power
exponential is generated by the polynomial applied to the sparse signal
itself
~\cite{METZLER20001,LISCHKE2020109009}.  In particular, the shifted
model attains the optimal fixed-scale exponent
$\epsilon^{-\frac{1}{2\alpha}}$ proved in Appendix~\ref{app:optimality}.

The shifted access assumption is concrete for the usual nearest-neighbor
finite-difference Laplacian.  On a $D$-dimensional grid of spacing
$\Delta x$, write
\begin{equation}
    \mathcal L
    =
    \frac{1}{(\Delta x)^2}
    \left(2D I-A_{\mathrm{nn}}\right),
    \label{eq:finite_difference_laplacian}
\end{equation}
where $A_{\mathrm{nn}}$ is the nearest-neighbor adjacency matrix.  Setting
$\beta_{\mathcal L}=4D/(\Delta x)^2$ gives the exact identity
\begin{equation}
    \frac{2\mathcal L}{4D/(\Delta x)^2}-I
    =
    -\frac{A_{\mathrm{nn}}}{2D}.
    \label{eq:laplacian_shifted_signal}
\end{equation}
For a periodic grid, $4D/(\Delta x)^2=\norm{\mathcal L}$ whenever the
highest Fourier mode is present, and every row on the right-hand side
of Eq.~\eqref{eq:laplacian_shifted_signal} has absolute sum one.  For
Dirichlet boundaries, the same number is a known upper bound on
$\norm{\mathcal L}$ and boundary rows have absolute sum at most one.
Because the matrix is $2D$-sparse and its largest entry has magnitude
$1/(2D)$, standard sparse-access block encoding has normalization
$2D\times1/(2D)=1$~\cite{Berry2007,Gilyen2019qsvt,
Sunderhauf2024blockencoding,Camps_2022}.
The shifted signal is also the negative transition matrix of the
nearest-neighbor random walk~\cite{1366222}.  Thus its oracle can be prepared
directly without first forming the generally dense fractional power
$\mathcal L^\alpha$.  For nonperiodic boundaries, the complexity
formulas use $4D/(\Delta x)^2$ in place of the exact spectral norm, as
specified after Eq.~\eqref{eq:normalized_target}.

At fixed $u_r$ and $\epsilon$, Eq.~\eqref{eq:heat_application_queries}
has shifted-oracle scale
$\Ocal(\sqrt{\kappa TD}/\Delta x)$.  The standard and shifted
biharmonic scales are, respectively,
$\Ocal(D\sqrt{\kappa T}/(\Delta x)^2)$ and
$\Ocal(\sqrt D(\kappa T)^{\frac{1}{4}}/\Delta x)$.  For noninteger $\alpha$,
the full grid-dependent terms in Eq.~\eqref{eq:fractional_application_queries}
are
\begin{equation}
\Ocal\!\left(
\frac{D u_r}{(\Delta x)^2}
\left[\frac{\kappa T u_r}{\epsilon}\right]^{\frac{1}{\alpha}}
\right),
\qquad
\Ocal\!\left(
\frac{\sqrt D u_r}{\Delta x}
\left[\frac{\kappa T u_r}{\epsilon}\right]^{\frac{1}{2\alpha}}
\right)
\label{eq:fractional_grid_queries}
\end{equation}
for standard and shifted access.  The Holtsmark stable index $3/2$
~\cite{zolotarev1986stable} corresponds here to $\alpha=3/4$; the
shifted expression then becomes
$\Ocal((\sqrt D/\Delta x)u_r^{\frac{5}{3}}
(\kappa T/\epsilon)^{\frac{2}{3}})$, reproducing the explicit mesh and
success-amplification dependence of the earlier Poisson-summation
construction.
Each result additionally uses $\Ocal(u_r)$ calls to the initial-state
oracle.

These statements concern the semidiscrete propagator.  A continuum
PDE guarantee must additionally allocate error to spatial
discretization, preparation of $u(0)$, and extraction of the desired
observable.  In particular, mesh refinement generally increases
$\norm{\mathcal L}$, so the matrix-function query bound alone does not
constitute an end-to-end exponential speedup for a continuum problem.

\subsection{Time-independent non-Hermitian simulation}
\label{subsec:nonhermitian}

Consider the time-independent generator
\begin{equation}
    A=L+\mathrm iG,
    \qquad
    L=L^\dagger\succeq0,
    \qquad
    G=G^\dagger.
    \label{eq:cartesian_decomposition}
\end{equation}
The access model supplies a block encoding of the full generator,
\begin{equation}
    (\bra{0^a}\otimes I)U_A(\ket{0^a}\otimes I)
    =\frac{A}{\beta_A},
    \qquad
    \beta_A\geq\norm{A}.
    \label{eq:nonhermitian_access}
\end{equation}
The adjoint circuit $U_A^\dagger$ block encodes
$A^\dagger/\beta_A$, and hence
\begin{equation}
    \frac{L}{\beta_A}
    =\frac12\left(\frac{A}{\beta_A}
    +\frac{A^\dagger}{\beta_A}\right),
    \qquad
    \frac{G}{\beta_A}
    =\frac{1}{2\mathrm i}\left(\frac{A}{\beta_A}
    -\frac{A^\dagger}{\beta_A}\right).
    \label{eq:cartesian_block_encodings}
\end{equation}
Each expression is a two-term LCU with coefficient one-norm one.
Thus both Hermitian blocks retain the single normalization $\beta_A$
and use only a constant number of calls to $U_A$ and $U_A^\dagger$
~\cite{Childs2012lcu,Gilyen2019qsvt,
dong2025productsblockencodings}.
Contour-integral, Laplace-transform, and eigenvalue-transformation
methods provide other block-encoded routes to non-Hermitian matrix
functions~\cite{Takahira_2020,takahira2021quantumalgorithmsbasedblockencoding,
jiang2026contourintegralbasedquantumeigenvalue,10756112,
an2024laplacetransformbasedquantum,chan2023simulatingnonunitarydynamicsusing,
wang2025quantumsimulationnonunitarydynamics}.  The construction below
instead preserves the Hamiltonian branches of LCHS\@.

The Weyl--Poisson construction uses two complementary ingredients.  Weyl
calculus lifts a scalar Fourier representation to the noncommuting line
of Hamiltonians $kL+G$
~\cite{ni2026quantumeigenvaluetransformationlinear}.  The scalar
surrogate is entire, while its Gaussian-regularized Fourier kernel is
analytic in a strip; these properties support the contour and alias
estimates below.  The kernel and query guarantee come from optimal
LCHS~\cite{low2025optimalquantumsimulationlinear}, while the sinh--sinh
map comes from the corresponding LCHS quadrature analysis
~\cite{aftab2026lchsmpf}.  Poisson summation exposes the continuous LCHS
term and every discretization alias in one exact operator formula.  It
changes neither the kernel nor the Hamiltonian simulation used in each
branch and therefore does not improve the LCHS query complexity.
For a fixed constant $c>0$, set
\begin{align}
    q(k)
    &=
    \frac{\mathrm e^{c-\mathrm i c k}}{\pi(1+k^2)}
    \exp\!\left[-\frac{k^2+1}{4\gamma^2}\right],
    \nonumber\\
    \phi(x)
    &=\int_{\mathbb R}q(k)\mathrm e^{-\mathrm ikx}\,\mathrm dk.
    \label{eq:optimal_lchs_kernel}
\end{align}
The second relation fixes the Fourier-sign convention used in the
Weyl identity below.  The entire function $\phi$ approximates
$\mathrm e^{-x}$ for $x\geq0$, but it is not equal to that function.
For noncommuting $L$ and $G$, the Weyl operator introduced below should
not be identified with the ordinary product
$\phi(TL)\mathrm e^{-\mathrm iTG}$ or with a holomorphic substitution in
$L+\mathrm iG$; it is defined by its Fourier integral along $kL+G$.
For $0<\epsilon\leq0.9$, choose
\begin{equation}
    \gamma=
    \frac{1}{c}
    \sqrt{c+\log\!\left(
    \frac{1+\frac{1}{2\pi}}{\epsilon}\right)},
    \qquad
    R=2c\gamma^2,
    \label{eq:optimal_lchs_parameters}
\end{equation}
The generalized LCHS theorem then gives
\begin{equation}
    \left\|
    \mathrm e^{-T(L+\mathrm iG)}-
    \int_{-R}^{R}q(k)
    \mathrm e^{-\mathrm iT(kL+G)}\,\mathrm dk
    \right\|
    \leq\epsilon.
    \label{eq:optimal_lchs_error}
\end{equation}
No commutativity of $L$ and $G$ is required.  The kernel also satisfies
\begin{equation}
    \int_{\mathbb R}\abs{q(k)}\,\mathrm dk
    =
    \mathrm e^c\operatorname{erfc}\!\left(\frac{1}{2\gamma}\right)
    \leq \mathrm e^c,
    \label{eq:optimal_lchs_normalization}
\end{equation}
Here $\operatorname{erfc}$ is the complementary error function, so the
continuous LCU normalization is constant when $c$ is constant.

For real $a$, the Fourier--Stieltjes Weyl calculus on the line of
Hamiltonians $kL+G$ reads
\begin{equation}
\operatorname{Op}^{\mathrm W}_{TL,TG}
\!\left[\mathrm e^{-\mathrm iy}\phi(x+a)\right]
:=
\int_{\mathbb R}q(k)\mathrm e^{-\mathrm iak}
\mathrm e^{-\mathrm iT(kL+G)}\,\mathrm dk.
\label{eq:weyl_fourier_class}
\end{equation}
Operator-valued Poisson summation now gives, for every $h>0$,
\begin{align}
&h\sum_{m\in\mathbb Z}
q(mh)\mathrm e^{-\mathrm iT(mhL+G)}
\nonumber\\
&\quad=
\operatorname{Op}^{\mathrm W}_{TL,TG}
\!\left[\mathrm e^{-\mathrm iy}\phi(x)\right]
\nonumber\\
&\qquad+
\sum_{\ell\ne0}
\operatorname{Op}^{\mathrm W}_{TL,TG}
\!\left[
\mathrm e^{-\mathrm iy}
\phi\!\left(x+\frac{2\pi\ell}{h}\right)
\right].
\label{eq:weyl_poisson_full}
\end{align}
Eq.~\eqref{eq:weyl_poisson_full} is the Weyl--Poisson formula used
here.  Its first term is exactly the continuous full-line LCHS integral,
and the remaining terms are its complete Poisson aliases.  The equality
is exact; only the identification of the first term with the target
$\mathrm e^{-T(L+\mathrm iG)}$ is approximate because $\phi$ is an
approximation to $\mathrm e^{-x}$.
The formula retains the noncommuting exponential
$\mathrm e^{-\mathrm iT(kL+G)}$.  It does not use a product formula or
replace $L$ and $G$ by commuting variables.  Truncating the left side
and isolating the target gives the exact error decomposition
\begin{align}
    &h\sum_{m=-N}^{N}
    q(mh)\mathrm e^{-\mathrm iT(mhL+G)}
    -\mathrm e^{-T(L+\mathrm iG)}
    \nonumber\\
    &=
    \left\{
    \operatorname{Op}^{\mathrm W}_{TL,TG}
    \!\left[\mathrm e^{-\mathrm iy}\phi(x)\right]
    -\mathrm e^{-T(L+\mathrm iG)}
    \right\}
    \nonumber\\
    &\quad+
    \sum_{\ell\ne0}
    \operatorname{Op}^{\mathrm W}_{TL,TG}
    \!\left[
        \mathrm e^{-\mathrm iy}
        \phi\!\left(x+\frac{2\pi\ell}{h}\right)
    \right]
    \nonumber\\
    &\quad-
    h\sum_{\abs{m}>N}
    q(mh)\mathrm e^{-\mathrm iT(mhL+G)}.
    \label{eq:weyl_poisson_exact}
\end{align}
The three terms are respectively the optimal-scaling LCHS kernel
mismatch, the complete Poisson alias, and the lattice tail.  The first term is present
because the selected optimal-scaling LCHS kernel has
$\phi\neq\mathrm e^{-x}$; replacing $\phi$ formally would correspond to
a different kernel.  Adding and subtracting the integral over $[-R,R]$
bounds its norm by the left-hand side of
Eq.~\eqref{eq:optimal_lchs_error} plus
$\int_{\abs{k}>R}\abs{q(k)}\,\mathrm dk$.  The latter is
$\Ocal(\epsilon)$ for the parameters in
Eq.~\eqref{eq:optimal_lchs_parameters}, so the full-line approximation
error is
$\Ocal(\epsilon)$.

For an overall operator error, each displayed tolerance may be replaced
by a fixed fraction of the desired error.  The triangle inequality then
combines the kernel, Poisson-alias, and lattice-tail bounds.  Such a
constant rescaling changes no asymptotic order below.

The uniform rule has $R=\Ocal(\log(1/\epsilon))$ and may take
$h^{-1}=\Ocal(T\beta_A+\log(1/\epsilon))$.  It therefore uses
\begin{equation}
    \Ocal\!\left(
        \log\frac{1}{\epsilon}
        \left[T\beta_A+\log\frac{1}{\epsilon}\right]
    \right)
    \label{eq:uniform_lchs_nodes}
\end{equation}
    nodes~\cite{low2025optimalquantumsimulationlinear,
    trefethen2014exponentially}.  The existing sinh--sinh quadrature can
    be combined with the same kernel~\cite{mori2005discovery,
    ooura2005double,sugihara1997optimality,okayama2022double,
    aftab2026lchsmpf}.  We use its unit-scale form $k=\sinh z$; inserting
a fixed positive scale changes only constants in the bounds below.
Poisson summation in $z$ gives the equally explicit compatibility
identity
\begin{align}
    &h\sum_{m\in\mathbb Z}
    \cosh(mh)q(\sinh(mh))
    \mathrm e^{-\mathrm iT(\sinh(mh)L+G)}
    \nonumber\\
    &=
    \int_{\mathbb R}q(k)
    \mathrm e^{-\mathrm iT(kL+G)}\,\mathrm dk
    \nonumber\\
    &\quad+
    \sum_{\ell\ne0}
    \int_{\mathbb R}
    \cosh z\,q(\sinh z)
    \mathrm e^{-\mathrm iT(\sinh zL+G)}
    \mathrm e^{-\frac{2\pi\mathrm i\ell z}{h}}\,\mathrm dz.
    \label{eq:sinh_weyl_poisson}
\end{align}
The first term on the right is the continuous LCHS integral with the
optimal-scaling kernel; the second contains every alias created by the
transformed lattice.
Thus the sinh--sinh change of variables is compatible with the same
Weyl--Poisson decomposition.  On the
real axis, the transformed kernel satisfies
\begin{equation}
    \cosh x\,\abs{q(\sinh x)}
    =
    \frac{\mathrm e^{c-\frac{1}{4\gamma^2}}}{\pi\cosh x}
    \exp\!\left[-\frac{\sinh^2x}{4\gamma^2}\right].
    \label{eq:sinh_optimal_kernel_decay}
\end{equation}
Its tail is double exponential in the transformed cutoff.  The mesh
must nevertheless resolve the growth of the complex-time Hamiltonian
on the displaced contour.  For fixed positive $T\beta_A$, or as this
physical scale grows, Appendix~\ref{app:weyl_poisson} gives
\begin{align}
    &2N+1=\Ocal\!\left(
        T\beta_A
        \log\frac{1}{\epsilon}
        \log\log\frac{1}{\epsilon}
    \right),
    \nonumber\\
    &\qquad \max_{\abs{m}\leq N}\abs{\sinh(mh)}
    =\Ocal\!\left(\log\frac{1}{\epsilon}\right).
    \label{eq:sinh_lchs_nodes_radius}
\end{align}
At fixed positive $T\beta_A$, the first bound has
$\Ocal(\log(1/\epsilon)\log\log(1/\epsilon))$ precision dependence.
As the supplied physical scale grows, however, the branch count retains
a linear $T\beta_A$ factor.  Thus the transformation
can reduce the precision-dominated branch count, but neither the largest
simulated frequency nor its physical-scale dependence.  Since no specific
state-preparation and controlled-selection (PREPARE/SELECT)
architecture is fixed here, this node reduction
is not claimed as an end-to-end gate-complexity improvement.
If a sharper upper bound on $\norm{L}$ is known separately, it may
replace $\beta_A$ when choosing the quadrature mesh.
Eq.~\eqref{eq:sinh_lchs_nodes_radius} is an upper bound for this
quadrature rule in the stated regime, not a lower bound on every
possible quadrature.

In contrast, the discrete LCU coefficient one-norm is uniformly
independent of $\beta_A$.  For $h\leq1$, the explicit estimate is
\begin{equation}
    h\sum_{m=-N}^{N}
    \cosh(mh)\abs{q(\sinh(mh))}
    \leq \mathrm e^c\left(1+\frac{1}{\pi}\right).
    \label{eq:sinh_lchs_normalization_bound}
\end{equation}
The right-hand side is also independent of $T$, $\epsilon$, and $N$.
The coefficient normalization and the number of coherent branches are
therefore distinct resources: only the former is norm independent.

For every sampled frequency, Eq.~\eqref{eq:cartesian_block_encodings}
gives a block encoding of
\begin{equation}
    \frac{kL+G}{\beta_A(\abs{k}+1)}.
    \label{eq:lchs_branch_normalization}
\end{equation}
The largest frequency is $R=\Ocal(\log(1/\epsilon))$.  Hamiltonian
simulation of the controlled branches and the constant LCU
normalization therefore give, for fixed positive $T\beta_A$ or growing
physical scale and under the resource convention stated at the
beginning of Sec.~\ref{sec:poisson_qsvt},
\begin{equation}
    \Ocal\!\left(
        T\beta_A\log\frac{1}{\epsilon}
    \right)
    \quad\text{queries to }U_A\text{ and }U_A^\dagger
    \label{eq:optimal_lchs_query_complexity}
\end{equation}
for operator-norm error at most $\epsilon$
~\cite{PhysRevLett.114.090502,
low2017hamiltoniansimulationuniformspectral,
low2025optimalquantumsimulationlinear}.  This is the optimal LCHS query
order expressed through the supplied normalization $\beta_A$.  Neither
the Weyl rewriting nor the sinh--sinh map changes this query order.

Preparing the normalized state with Euclidean error at most $\epsilon$
requires
\begin{equation}
    \Ocal\!\left(
        u_rT\beta_A\log\frac{u_r}{\epsilon}
    \right)
    \label{eq:optimal_lchs_state_complexity}
\end{equation}
queries to $U_A$ and $U_A^\dagger$, together with $\Ocal(u_r)$ calls
to the initial-state oracle.  The corresponding number of branches in
each LCHS block is
\begin{equation}
    \Ocal\!\left(
        T\beta_A
        \log\frac{u_r}{\epsilon}
        \log\log\frac{u_r}{\epsilon}
    \right).
    \label{eq:sinh_lchs_state_branches}
\end{equation}
At fixed $T\beta_A$ this again reduces to the precision-only
$\log\log$ compression.  The reduction concerns the node-dependent
    PREPARE/SELECT
description; it is not, by itself, a query or gate-complexity speedup.
It does not remove the constant LCU normalization or improve
Eq.~\eqref{eq:optimal_lchs_state_complexity}.  The Hermitian result of
Sec.~\ref{sec:poisson_qsvt} remains the case in which the frequency LCU
is eliminated entirely.

\subsection{Amplitude--phase separation}
\label{subsec:aps}

Amplitude--phase separation (APS) separates the dissipative factor from
the remaining interaction in non-unitary simulation.  This viewpoint is
also relevant to quantum algorithms for open-system dynamics
~\cite{RevModPhys.93.015008,Liu2025simulationofopen,
delgado2025quantum}.  We consider the time-independent generator
\begin{equation}
    A=H+\mathrm iG,
    \qquad H=H^\dagger\succeq0,
    \qquad G=G^\dagger,
    \label{eq:aps_generator}
\end{equation}
Let $\beta_H\geq\norm{H}$ and $\beta_G\geq\norm{G}$.  We assume a
block encoding of $G/\beta_G$ and, for a fixed positive integer
$\alpha$, an exact half-root oracle whose designated block is
\begin{equation}
    \left(\frac{H}{\beta_H}\right)^{\frac{1}{2\alpha}}.
    \label{eq:aps_half_root_access}
\end{equation}
Applying the degree-two QSVT polynomial $2x^2-1$ also gives the
root-shift signal
\begin{equation}
    2\left(\frac{H}{\beta_H}\right)^{\frac{1}{\alpha}}-I
    =
    \frac{2H^{\frac{1}{\alpha}}}{\beta_H^{\frac{1}{\alpha}}}-I,
    \label{eq:aps_root_signal}
\end{equation}
with two half-root queries and no approximation error
~\cite{Gilyen2019qsvt}.  Here $\alpha$ labels the order of the available
root access; it does not change the physical generator $H$ or the target
propagator $\mathrm e^{-TH}$.  The fractional powers are defined by the
spectral calculus.  For $\alpha=1$,
Eq.~\eqref{eq:aps_half_root_access} is the square-root access of
APS~\cite{hu2026quantumsimulationnonunitarydynamics}; for a fixed integer
$\alpha>1$ it is the corresponding root of order $2\alpha$.  Such access
is not implied by a generic block encoding of $H$.

APS gives two exact interaction-picture factorizations
~\cite{hu2026quantumsimulationnonunitarydynamics}.  Separating the phase
first gives
\begin{equation}
    \mathrm e^{-TA}
    =
    \mathrm e^{-\mathrm iTG}
    \mathcal T
    \exp\!\left[
        -\int_0^T
        \mathrm e^{\mathrm isG}H
        \mathrm e^{-\mathrm isG}\,\mathrm ds
    \right],
    \label{eq:phase_driven_aps}
\end{equation}
whereas separating the amplitude first gives
\begin{equation}
    \mathrm e^{-TA}
    =
    \mathrm e^{-T H}
    \mathcal T
    \exp\!\left[
        -\mathrm i\int_0^T
        \mathrm e^{sH}G
        \mathrm e^{-sH}\,\mathrm ds
    \right].
    \label{eq:amplitude_driven_aps}
\end{equation}
No commutativity between $H$ and $G$ is used in either equality.  The
similarity-transformed generator
$\mathrm e^{sH}G\mathrm e^{-sH}$ need not be Hermitian, so the second
factor in Eq.~\eqref{eq:amplitude_driven_aps} is not generally unitary.

The reusable structure in Eq.~\eqref{eq:amplitude_driven_aps} is the
time-independent contraction semigroup
\begin{equation}
    \mathrm e^{-sH}\mathrm e^{-tH}
    =\mathrm e^{-(s+t)H},
    \qquad s,t\geq0.
    \label{eq:aps_semigroup}
\end{equation}
After expanding the interaction factor, every dissipative interval is
nonnegative and the intervals in each Dyson product sum to $T$.  The
algorithm therefore requires one uniformly controlled semigroup family,
rather than unrelated approximations at each time.  Since
$H=(H^{\frac{1}{2\alpha}})^{2\alpha}$,
Eq.~\eqref{eq:aps_half_root_access} is the standard positive signal for
an even power $2\alpha$, with base normalization
$\beta_H^{\frac{1}{2\alpha}}$.  The Poisson-compiled polynomial
therefore provides an ordinary even-QSVT realization of this factor.
The amplitude-driven APS algorithm also queries the
controlled family
\begin{equation}
    \sum_m\ket{m}\!\bra{m}\otimes
    \mathrm e^{-t_mH},
    \qquad 0\leq t_m\leq T.
    \label{eq:aps_controlled_family}
\end{equation}
This family has a direct Poisson-compiled realization.  For each $t_m$,
synthesize the corresponding even-QSVT sequence and choose one common
maximum query length.  Shorter sequences are padded in steps of two.
Replacing
each phase rotation by its multiplexed version, controlled by $\ket m$,
gives one block-diagonal circuit whose $m$th block implements
$\mathrm e^{-t_mH}$.  The signal calls are common to all branches; only
the phase angles depend on $m$.  Appendix~\ref{app:aps} proves this
direct-sum construction and its uniform error bound
~\cite{Gilyen2019qsvt}.
For a per-call error $\delta$, the scalar polynomial is computed with
smaller error and rescaled only by $1+\delta/2$.  Because it is even and
bounded by one, the real-polynomial QSVT theorem implements it with
normalization one~\cite{Gilyen2019qsvt}.  Each designated block is
therefore a normalization-one approximation to
$\mathrm e^{-t_mH}$.  The fixed normalization two used for standalone
state preparation in Theorem~\ref{thm:poisson_compilation} is not
repeated inside the APS products.  Direct root-shift access remains
sufficient for the standalone shifted-oracle results.  For repeated APS
products, half-root access avoids the degree-dependent normalization of
a mixed-parity generalized eigenvalue transformation.

We use Eq.~\eqref{eq:aps_controlled_family} as the dissipative-family
oracle in the APS construction
~\cite{hu2026quantumsimulationnonunitarydynamics}.  One realization of
this family has a maximum QSVT degree, while the truncated Dyson
construction invokes the family and the block encoding of $G/\beta_G$
multiple times.  We give the query count both when the controlled
family is provided directly and when it is implemented from the
half-root oracle.  In the latter access model, every family invocation
executes the corresponding QSVT sequence.
The case $G=0$ reduces to the standalone dissipative transformation.

\begin{corollary}[Poisson-compiled amplitude-driven APS]
\label{cor:aps_complexity}
Let $0<\epsilon<1/2$ and $\beta_G T\geq1$.  Under the preceding access
assumptions, one invocation of the controlled family in
Eq.~\eqref{eq:aps_controlled_family}, with the uniform accuracy required
by the full APS circuit, can be realized using
\begin{equation}
\Ocal\!\left(
 \left(
 T\beta_H+\log\frac{u_r\beta_G T}{\epsilon}
 \right)^{\frac{1}{2\alpha}}
 \log^{1-\frac{1}{2\alpha}}\!\frac{u_r\beta_G T}{\epsilon}
\right)
\label{eq:aps_family_degree}
\end{equation}
half-root queries.  To prepare a normalized state proportional to
$\mathrm e^{-T(H+\mathrm iG)}u_0$ with Euclidean error at most
$\epsilon$ and constant success probability, amplitude-driven APS uses
\begin{equation}
\Ocal\!\left(
u_r\beta_G T
\frac{\log(u_r/\epsilon)}{\log\log(u_r/\epsilon)}
\right)
\label{eq:aps_family_calls}
\end{equation}
invocations of the controlled family and the same order of calls to
the block encoding of $G/\beta_G$.  Consequently, when the controlled
family is expanded into its QSVT realization, the primitive query
counts are
\begin{align}
\mathcal Q_{\mathrm{half}}
={}&\Ocal\!\Bigg(
u_r\beta_G T
\frac{\log(u_r/\epsilon)}{\log\log(u_r/\epsilon)}
\nonumber\\[-1mm]
&\quad\times
\left(
T\beta_H+\log\frac{u_r\beta_G T}{\epsilon}
\right)^{\frac{1}{2\alpha}}
\log^{1-\frac{1}{2\alpha}}\!\frac{u_r\beta_G T}{\epsilon}
\Bigg),
\label{eq:aps_flat_half_root_queries}\\
\mathcal Q_G
={}&\Ocal\!\left(
u_r\beta_G T
\frac{\log(u_r/\epsilon)}{\log\log(u_r/\epsilon)}
\right).
\label{eq:aps_flat_g_queries}
\end{align}
The algorithm also uses $\Ocal(u_r)$ calls to the initial-state oracle.
\end{corollary}

The factorization and error allocation behind this corollary are given
in Appendix~\ref{app:aps}.  Telescoping controls every controlled-family
invocation and the amplitude-amplification steps.  The lower-order
iterated logarithms introduced by this allocation are absorbed into
$\log(u_r\beta_G T/\epsilon)$.

If
$T\beta_H\geq\log(u_r\beta_G T/\epsilon)$, then for $\alpha=1$
the per-invocation bound in Eq.~\eqref{eq:aps_family_degree} reduces to
\begin{equation}
\Ocal\!\left(
\sqrt{T\beta_H\log\frac{u_r\beta_G T}{\epsilon}}
\right).
\label{eq:aps_square_root_reduction}
\end{equation}
This reproduces the square-root dependence of the dissipative family in
the amplitude-driven APS framework
~\cite{hu2026quantumsimulationnonunitarydynamics}, using the
Poisson-compiled QSVT realization.  Eq.~\eqref{eq:aps_family_calls} is the separate interaction-picture cost.
Combining Eqs.~\eqref{eq:aps_square_root_reduction}
and~\eqref{eq:aps_family_calls} gives the elementary half-root query
count in Eq.~\eqref{eq:aps_flat_half_root_queries}.  At fixed $T\beta_H$,
$\beta_G T$, and $u_r$, the sharper interpolation formula in
Eq.~\eqref{eq:shifted_integer_degree_interpolation} improves the
per-invocation high-precision factor to
$\log(u_r\beta_G T/\epsilon)/\log\log(u_r\beta_G T/\epsilon)$.

In a discretized phase-driven APS
algorithm, each conjugated dissipative step has the form
\begin{equation}
    \mathrm e^{-\Delta t\,
    \mathrm e^{\mathrm isG}H\mathrm e^{-\mathrm isG}}
    =
    \mathrm e^{\mathrm isG}
    \mathrm e^{-\Delta t H}
    \mathrm e^{-\mathrm isG},
    \label{eq:aps_conjugated_step}
\end{equation}
so the same polynomial-transform primitive applies whenever controlled
variable-duration calls are available.

At fixed precision, the dissipative term in
Eq.~\eqref{eq:aps_family_degree} scales as
$(T\beta_H)^{\frac{1}{2\alpha}}$.  The case $\alpha=1$ reproduces the
square-root APS scaling, whereas every fixed integer $\alpha>1$ gives
beyond-square-root dependence on the same physical scale $T\beta_H$.
This statement concerns the realization of each controlled dissipative
family call; combining it with the family-call count gives
Eq.~\eqref{eq:aps_flat_half_root_queries}.  Fourier kernels beyond the
Gaussian case may change sign, but this does not obstruct the present
construction: their samples are summed classically into one QSVT
polynomial and are never used as probabilities or coherent LCU
coefficients.  The conclusion is conditional on the half-root access in
Eq.~\eqref{eq:aps_half_root_access}.  The degree-two relation in
Eq.~\eqref{eq:aps_root_signal} shows how this access also supplies the
root-shift signal.  This is consistent with general
limitations on fast-forwarding unstructured linear differential
equations~\cite{An2025qq}.

\section{Discussion and Conclusion}
\label{sec:conclusion}

Poisson summation is used here as a classical compiler for
$\mathrm e^{-T H^\alpha}$.  Its exact target--alias--tail identity
controls the sampled coefficients, which are combined before one
polynomial-transform circuit is synthesized.  The quantum success
amplitude therefore depends on the input-to-output ratio $u_r$, but not
on a Fourier coefficient one-norm.

The encoded signal determines the approximation geometry.  The direct
even extension under $H/\norm{H}$ is entire for even positive integers
and has an interior algebraic singularity otherwise.  Under
$2H/\norm{H}-I$, the exact quadratic lift makes every positive integer
entire and doubles the fixed-scale approximation exponent for
noninteger powers.  The matching bounds and their scope are summarized
in Table~\ref{tab:degree_regimes}.  Standard finite-difference
Laplacians provide a concrete unit-normalized realization of the
shifted oracle.

In the time-independent non-Hermitian application, the Fourier--Weyl
representation of the entire scalar surrogate leads to an exact
operator-valued Poisson alias identity.  The Fourier representation
defines the noncommuting operator lift, while strip analyticity and
decay of its kernel control the contour and alias estimates.  This
construction remains compatible with
sinh--sinh discretization but changes neither the kernel, the LCU
normalization, nor the optimal LCHS query order.  The sinh--sinh map can
lower the precision-dominated node count, but no end-to-end query or
gate-complexity speedup is claimed.  The LCU coefficient one-norm is
independent of $\beta_A$, while the uniform branch count retains its
$T\beta_A$ dependence.

The APS extension uses a different structure: the semigroup law in
Eq.~\eqref{eq:aps_semigroup} makes one uniformly controlled family of
dissipative contractions reusable throughout the Dyson products.  Its
implementation is based on the half-root access in
Eq.~\eqref{eq:aps_half_root_access}, whose degree-two QSVT image is the
root-shift signal in Eq.~\eqref{eq:aps_root_signal}.  The construction
recovers the square-root degree when $\alpha=1$; combining the
per-invocation degree with the interaction-picture family calls gives
the elementary half-root query count.

The complexity results count block-encoding queries.  Kernel
evaluation, Poisson sampling, and QSP phase synthesis are classical
preprocessing costs and are not identified with fault-tolerant gate
complexity.  The present analysis specifies their sample counts and
required numerical precision, but not an optimized bit complexity for
general $\alpha$.  Open approximation questions include uniform joint
time--precision bounds for general non-even powers.  The non-even lower
bound under standard access is complete for ordinary single-sequence
QSVT, but does not extend to generalized indefinite-parity
transformations.


\bibliography{ref}

\clearpage
\onecolumngrid
\appendix
\section{Fourier-kernel estimates and Poisson truncation}
\label{app:kernel}

We first establish the Fourier decay and truncation bounds underlying
the Poisson-compiled polynomial.  Even powers have a
stretched-exponential Fourier tail, whereas every non-even power
inherits an algebraic tail from the singularity at the origin.  The
same estimates also justify the absolute convergence needed for the
exact Poisson identity.

\begin{lemma}[Poisson residual]
\label{lem:poisson_residual}
For every $\alpha,T>0$, the function
$\mathrm e^{-T\norm{H}^{\alpha}\abs{x}^{\alpha}}$ and the kernel
$f_{\alpha,T}$ are integrable and continuous.
Eqs.~\eqref{eq:poisson_exact} and
\eqref{eq:poisson_alias_tail} hold for every $h>0$; both infinite
series in Eq.~\eqref{eq:poisson_exact} converge absolutely and
uniformly for $x$ in a compact interval.
\end{lemma}

\begin{proof}
Write the right-hand side of Eq.~\eqref{eq:poisson_exact} as the
periodization of the target with period $1/h$.  Exponential decay of
the target makes this periodization continuous, and its defining
series converges uniformly on every compact interval.  Its $k$th
Fourier coefficient on any interval of length $1/h$ is
\begin{align*}
h\int_{-\frac{1}{2h}}^{\frac{1}{2h}}
\sum_{n\in\mathbb Z}
\mathrm e^{-T\norm{H}^{\alpha}\abs{x+n/h}^{\alpha}}
\mathrm e^{2\pi\mathrm i khx}\,\mathrm dx
=
h\int_{\mathbb R}
\mathrm e^{-T\norm{H}^{\alpha}\abs{x}^{\alpha}}
\mathrm e^{2\pi\mathrm i khx}\,\mathrm dx
=h f_{\alpha,T}(kh).
\end{align*}
The interchange of sum and integral follows from absolute
convergence.  Lemma~\ref{lem:kernel_decay} implies
$\sum_k\abs{f_{\alpha,T}(kh)}<\infty$ for every $h>0$.  Hence the
Fourier series with these coefficients converges absolutely and
uniformly to the periodization, proving
Eq.~\eqref{eq:poisson_exact} without a distributional limiting
argument~\cite{katznelson2004introduction}.  Isolating the $n=0$ term
and the terms $\abs{k}\leq N$ proves
Eq.~\eqref{eq:poisson_alias_tail}.
\end{proof}

\begin{lemma}[Kernel decay]
\label{lem:kernel_decay}
If $\alpha$ is an even positive integer, there are positive constants
$c_\alpha,C_\alpha$ such that
Eq.~\eqref{eq:even_kernel_decay} holds.  If $\alpha$ is not an even
integer, Eq.~\eqref{eq:noneven_kernel_asymptotic} holds and, for
$\abs{y}\geq C_\alpha T^{\frac{1}{\alpha}}\norm{H}$,
\begin{equation}
    \abs{f_{\alpha,T}(y)}
    \leq
    C_\alpha
    \frac{T\norm{H}^{\alpha}}{\abs{y}^{\alpha+1}}.
    \label{eq:noneven_kernel_bound}
\end{equation}
\end{lemma}

\begin{proof}
Suppose first that $\alpha$ is even.  For $y>0$, translate the
integration variable in Eq.~\eqref{eq:inverse_fourier_kernel} by
$\mathrm is$, where $s>0$.  Since the integrand is entire and
decays on the joining segments,
\[
    \abs{f_{\alpha,T}(y)}
    \leq
    \mathrm e^{-2\pi ys}
    \int_{\mathbb R}
    \exp\!\left[
       -T\norm{H}^{\alpha}
       \operatorname{Re}(\xi+\mathrm is)^\alpha
    \right]\mathrm d\xi.
\]
The elementary polynomial estimate
\[
    \operatorname{Re}(\xi+\mathrm is)^\alpha
    \geq
    c_\alpha\abs{\xi}^{\alpha}
    -
    C_\alpha s^\alpha
\]
follows from homogeneity and Young's inequality.  It shows that the
integral is at most
\[
    \frac{C_\alpha}{T^{\frac{1}{\alpha}}\norm{H}}
    \exp(C_\alpha T\norm{H}^{\alpha}s^\alpha).
\]
Choosing
\[
    s=c_\alpha
    \left[
    \frac{y}{T\norm{H}^{\alpha}}
    \right]^{\frac{1}{\alpha-1}}
\]
proves
Eq.~\eqref{eq:even_kernel_decay}; evenness handles $y<0$.

If $\alpha$ is not even, choose a smooth cutoff equal to one near the
origin.  On its support,
\[
    \mathrm e^{-T\norm{H}^{\alpha}\abs{\xi}^{\alpha}}
    =1-T\norm{H}^{\alpha}\abs{\xi}^{\alpha}
    +\Ocal(T^2\norm{H}^{2\alpha}\abs{\xi}^{2\alpha}).
\]
The cutoff constant has a rapidly decreasing Fourier transform.  The
Fourier transform of the complement is also rapidly decreasing after
repeated integration by parts, since it is smooth and all of its
derivatives are integrable.  The first non-smooth local term therefore
determines the leading asymptotic.  Abel regularization gives
\[
    \int_0^\infty
    \xi^\alpha\cos(2\pi y\xi)\,\mathrm d\xi
    =
    -\frac{\Gamma(\alpha+1)\sin(\pi\alpha/2)}
    {(2\pi\abs{y})^{\alpha+1}}
\]
for $y\ne0$.  The localized Fourier version of Watson's lemma shows
that the remainder is $o(\abs{y}^{-\alpha-1})$~\cite{wong2001asymptotic}.
This proves Eq.~\eqref{eq:noneven_kernel_asymptotic}.  Finally, the
change of variables
$\xi\mapsto\xi/[T^{\frac{1}{\alpha}}\norm{H}]$ reduces the kernel to a fixed
profile.  Its asymptotic is therefore a uniform upper bound once
$\abs{y}\geq C_\alpha T^{\frac{1}{\alpha}}\norm{H}$, which proves
Eq.~\eqref{eq:noneven_kernel_bound}.
\end{proof}

For $\abs{x}\leq1$ and $h^{-1}>1$, the alias term obeys
\begin{equation}
    \sum_{n\ne0}
    \mathrm e^{-T\norm{H}^{\alpha}\abs{x+n/h}^{\alpha}}
    \leq
    2\sum_{n=1}^{\infty}
    \mathrm e^{-T\norm{H}^{\alpha}(n/h-1)^\alpha}.
    \label{eq:alias_bound}
\end{equation}
For a positive decreasing function, comparison with its integral gives
\begin{align}
\sum_{n=1}^{\infty}
\mathrm e^{-T\norm{H}^{\alpha}(n/h-1)^\alpha}
\leq
\mathrm e^{-T\norm{H}^{\alpha}(1/h-1)^\alpha}
\nonumber+h\int_{1/h-1}^{\infty}
\mathrm e^{-T\norm{H}^{\alpha}y^\alpha}\,\mathrm dy.
\label{eq:alias_integral_comparison}
\end{align}
After the substitution
$y=T^{-\frac{1}{\alpha}}\norm{H}^{-1}z$, the last integral is an incomplete
gamma tail and is bounded by its first exponential times an
$\alpha$-dependent polynomial factor.  Absorbing that factor by
increasing the logarithm, the alias term and the sampled-kernel tail
are at most $\epsilon/[C\log(d+2)]$ when
\begin{equation}
    \frac{1}{h}
    \geq
    2+
    C_\alpha
    \left[
    \frac{1}{T\norm{H}^{\alpha}}
    \log\frac{C_\alpha\log(d+2)}{\epsilon}
    \right]^{\frac{1}{\alpha}}
    \label{eq:sample_spacing}
\end{equation}
and, for even integer $\alpha$,
\begin{equation}
    Nh
    \geq
    C_\alpha\norm{H}T^{\frac{1}{\alpha}}
    \left[
    \log\frac{C_\alpha\log(d+2)}{\epsilon}
    \right]^{1-\frac{1}{\alpha}}.
    \label{eq:even_frequency_radius}
\end{equation}
If $\alpha$ is not an even integer, it is sufficient to choose
\begin{equation}
    Nh
    \geq
    C_\alpha\norm{H}
    \left[
    \frac{T\log(d+2)}{\epsilon}
    \right]^{\frac{1}{\alpha}}.
    \label{eq:noneven_frequency_radius}
\end{equation}
These parameters are used only to compute the coefficients of $P_d$
classically.  Combining Eqs.~\eqref{eq:sample_spacing} and
\eqref{eq:even_frequency_radius} shows that an even power needs
\begin{equation}
N=\Ocal\!\left(
\norm{H}T^{\frac{1}{\alpha}}
\log^{1-\frac{1}{\alpha}}\frac{C_\alpha\log(d+2)}{\epsilon}
+\log\frac{C_\alpha\log(d+2)}{\epsilon}
\right)
\label{eq:even_classical_samples}
\end{equation}
kernel samples up to constant factors.  This quantity is a classical
preprocessing resource and is not the QSVT degree.
For a non-even power the same multiplication of the radius by
$h^{-1}$ gives
\begin{align}
N=\Ocal\!\Bigg(&
\norm{H}\left[
\frac{T}{\epsilon}
\log\frac{C_\alpha\log(d+2)}{\epsilon}
\right]^{\frac{1}{\alpha}}
+\left[
\frac{1}{\epsilon}
\log^2\frac{C_\alpha\log(d+2)}{\epsilon}
\right]^{\frac{1}{\alpha}}\Bigg).
\label{eq:noneven_classical_samples}
\end{align}

For shifted access, the classical sample counts follow from the same
two formulas after replacing $\alpha$ by $2\alpha$ and $\norm{H}$ by
$\sqrt{\norm{H}}$.  Thus Eq.~\eqref{eq:even_classical_samples}
applies to every integer $\alpha$, whereas
Eq.~\eqref{eq:noneven_classical_samples} applies only to noninteger
$\alpha$ after this replacement.  These counts remain classical
coefficient-generation costs and do not multiply the query bounds in
Eqs.~\eqref{eq:shifted_integer_degree} and
\eqref{eq:shifted_fractional_degree}.

\section{Proof of the Poisson-compiled polynomial bounds}
\label{app:poisson_compilation}

Let the degree-$d$ Chebyshev projection be denoted locally by
$S_d$.  Eq.~\eqref{eq:jacobi_anger_cosine} and linearity give
$P_d=S_d F_{N,h}$, where $F_{N,h}$ is the finite sum in
Eq.~\eqref{eq:cosine_poisson_sum}.  The projection is the ordinary
Fourier partial-sum operator after $x=\cos\theta$; hence its Lebesgue
constant obeys
\begin{equation}
    \norm{S_d}_{C[-1,1]\to C[-1,1]}
    \leq C\log(d+2)
    \label{eq:chebyshev_projection_norm}
\end{equation}
\cite{mason2003chebyshev,trefethen2019approximation}.  It follows that
\begin{align}
\norm{P_d-\mathrm e^{-T\norm{H}^{\alpha}\abs{x}^{\alpha}}}_\infty
\leq C\log(d+2)
\norm{F_{N,h}-\mathrm e^{-T\norm{H}^{\alpha}\abs{x}^{\alpha}}}_\infty
+
\norm{S_d\mathrm e^{-T\norm{H}^{\alpha}\abs{x}^{\alpha}}
-\mathrm e^{-T\norm{H}^{\alpha}\abs{x}^{\alpha}}}_\infty.
\label{eq:projection_error_split}
\end{align}
The parameter choices in
Eqs.~\eqref{eq:sample_spacing}--\eqref{eq:noneven_frequency_radius},
with $C_\alpha$ enlarged if necessary, make the first term at most
$\epsilon/4$.  We now bound the second term.

Suppose first that $\alpha$ is even.  The function
$\mathrm e^{-T\norm{H}^{\alpha}z^\alpha}$ is entire.  A point of the
Bernstein ellipse $E_{\mathrm e^s}$ has imaginary part at most
$\sinh s$.  For real $u,v$ and even $\alpha$,
$-\operatorname{Re}(u+\mathrm iv)^\alpha\leq C_\alpha\abs{v}^\alpha$;
this follows by dividing by $\abs{v}^\alpha$ and observing that the
negative part of the resulting polynomial is bounded.  For
$0<s\leq1$ this proves the first estimate below.  For $s\geq1$, the
elementary bound $\abs{z}\leq\mathrm e^s$ proves the second:
\begin{equation}
    \max_{z\in E_{\mathrm e^s}}
    \abs{\mathrm e^{-T\norm{H}^{\alpha}z^\alpha}}
    \leq
    \begin{cases}
    \exp(C_\alpha T\norm{H}^{\alpha}s^\alpha),
        &0<s\leq1,\\
    \exp(C_\alpha T\norm{H}^{\alpha}\mathrm e^{\alpha s}),
        &s\geq1.
    \end{cases}
    \label{eq:ellipse_maximum}
\end{equation}
The Bernstein-ellipse coefficient bound therefore makes the
degree-$d$ tail no larger than the corresponding line of
Eq.~\eqref{eq:ellipse_maximum} multiplied by
$4\mathrm e^{-ds}/(1-\mathrm e^{-s})$~\cite{trefethen2019approximation}.
Writing $\log(1/\epsilon)$ for the required exponential
suppression, choose $s$ proportional to
$[\log(1/\epsilon)/(T\norm{H}^{\alpha})]^{\frac{1}{\alpha}}$ when
this number is at most one.  The resulting sufficient degree is
\[
    d=\Ocal\!\left(
    \norm{H}T^{\frac{1}{\alpha}}
    \log^{1-\frac{1}{\alpha}}\frac{1}{\epsilon}
    \right),
\]
    including the additional logarithmic factor from
$(1-\mathrm e^{-s})^{-1}$.  When that choice of $s$ exceeds one, take
$s$ proportional to
$\log[1+(\log(1/\epsilon)/(T\norm{H}^{\alpha}))^{\frac{1}{\alpha}}]$.
The large-ellipse estimate then gives
\[
    d=
    \Ocal\!\left(
    \frac{\log(1/\epsilon)}{
    \log\!\left[
    1+
    \left(
    \frac{\log(1/\epsilon)}
    {T\norm{H}^{\alpha}}
    \right)^{\frac{1}{\alpha}}
    \right]}
    \right).
\]
Combining the two regimes gives the sharper joint estimate
\begin{equation}
\begin{split}
d=\Ocal\!\Bigg(
\norm{H}T^{\frac{1}{\alpha}}
\log^{1-\frac{1}{\alpha}}\frac{1}{\epsilon}
+\frac{\log(1/\epsilon)}{
\log\!\left[
1+
\left(
\frac{\log(1/\epsilon)}{T\norm{H}^{\alpha}}
\right)^{\frac{1}{\alpha}}
\right]}
\Bigg).
\end{split}
\label{eq:even_degree_interpolation}
\end{equation}
The denominator is understood under the high-precision convention at
the beginning of Sec.~\ref{sec:poisson_qsvt}.  Bounding it from below by one
proves the cleaner uniform form in Eq.~\eqref{eq:even_degree}.  At
fixed $T\norm{H}^{\alpha}$ it is
$\Theta(\log\log(1/\epsilon))$, which gives the stated
high-precision refinement.

Applying the same estimate with exponent $2\alpha$ and normalization
$\sqrt{\norm{H}}$, before the quadratic lifting, gives
\begin{equation}
\begin{split}
d=\Ocal\!\Bigg(
\sqrt{\norm{H}}T^{\frac{1}{2\alpha}}
\log^{1-\frac{1}{2\alpha}}\frac{1}{\epsilon}
+\frac{\log(1/\epsilon)}{
\log\!\left[
1+
\left(
\frac{\log(1/\epsilon)}{T\norm{H}^{\alpha}}
\right)^{\frac{1}{2\alpha}}
\right]}
\Bigg),
\end{split}
\label{eq:shifted_integer_degree_interpolation}
\end{equation}
which implies Eq.~\eqref{eq:shifted_integer_degree}.

If $\alpha$ is not even, the Chebyshev coefficients are obtained by
writing $x=\cos\theta$.  The only interior singularity occurs at
$\theta=\pi/2$, and
\[
    \mathrm e^{-T\norm{H}^{\alpha}
    \abs{\cos\theta}^{\alpha}}
    =
    1-T\norm{H}^{\alpha}
    \abs{\theta-\pi/2}^{\alpha}
    +o(\abs{\theta-\pi/2}^{\alpha}).
\]
Localizing this cosine-Fourier integral at $\pi/2$ and applying the
same Watson-lemma estimate as in Lemma~\ref{lem:kernel_decay} gives,
uniformly for
$j\geq C_\alpha\max\{1,\norm{H}T^{\frac{1}{\alpha}}\}$, the coefficient at index
$j$ bounded by
$C_\alpha T\norm{H}^{\alpha}j^{-\alpha-1}$.  The terms of higher order
in the local exponential form a convergent geometric majorant in this
range, while repeated integration by parts controls the complement.
This is also a standard consequence of the direct theorem for
algebraic singularities~\cite{wong2001asymptotic,
devore1993constructive,ditzian1987moduli}.  Summing the coefficient
tail gives
\begin{equation}
\norm{S_d\mathrm e^{-T\norm{H}^{\alpha}\abs{x}^{\alpha}}
-\mathrm e^{-T\norm{H}^{\alpha}\abs{x}^{\alpha}}}_\infty
\leq C_\alpha T\norm{H}^{\alpha}d^{-\alpha}.
\label{eq:noneven_chebyshev_tail}
\end{equation}
Choosing $d$ as in Eq.~\eqref{eq:noneven_degree} completes the
approximation estimate.

The preceding argument assumes exact kernel samples.  It is sufficient
to compute each $f_{\alpha,T}(kh)$ to additive error
$\epsilon/[C h(2N+1)\log(d+2)]$.  Indeed, the resulting finite-sum
error is at most $h(2N+1)$ times the sample error, and
Eq.~\eqref{eq:chebyshev_projection_norm} supplies the remaining
factor.  Rounding the final $d+1$ Chebyshev coefficients contributes
at most the sum of their absolute rounding errors.  Both are classical
precision requirements and do not alter the block-encoding query
count.

Finally, Eq.~\eqref{eq:poisson_polynomial_error} implies
$\abs{P_d(x)}\leq1+\epsilon/2$ on $[-1,1]$.
The rescaling in Eq.~\eqref{eq:qsvt_rescaling} is real, even, and
bounded by $1/2$.  The standard QSVT completion theorem therefore
produces a valid phase sequence of degree $d$~\cite{Gilyen2019qsvt}.
Writing
$g(x)=\mathrm e^{-T\norm{H}^{\alpha}\abs{x}^{\alpha}}$, the scalar
error satisfies
\begin{equation}
\left|
\frac{P_d(x)}{1+\epsilon/2}-g(x)
\right|
\leq
\frac{\epsilon/2+(\epsilon/2)g(x)}{1+\epsilon/2}
\leq\epsilon.
\label{eq:normalization_two_scalar_bound}
\end{equation}
The spectral theorem proves
Eq.~\eqref{eq:qsvt_normalization_two_error}.  This proves
Theorem~\ref{thm:poisson_compilation} and the $U_H$ bounds in
Theorem~\ref{thm:degree}.  Applying the same estimates with exponent
$2\alpha$ and using Proposition~\ref{prop:quadratic_lifting} proves
the $U_S$ bounds.

\section{End-to-end error propagation}
\label{app:end_to_end}

For nonzero vectors $v$ and $w$ with
$\norm{v-w}\leq\norm{w}/2$, the triangle inequality gives
\begin{equation}
    \left\|
    \frac{v}{\norm{v}}-\frac{w}{\norm{w}}
    \right\|
    \leq \frac{2\norm{v-w}}{\norm{w}}.
    \label{eq:normalization_lemma}
\end{equation}
On the normalized input $u_0/\norm{u_0}$, the exact output of the
normalization-two block is
$u_T/(2\norm{u_0})$ and has norm $1/(2u_r)$.  If the implemented block
produces $\widetilde u_T$, Eq.~\eqref{eq:normalization_lemma} gives
\begin{equation}
    \left\|
    \frac{\widetilde u_T}{\norm{\widetilde u_T}}
    -
    \frac{u_T}{\norm{u_T}}
    \right\|
    \leq
    4u_r
    \left\|
    \widetilde u_T-\frac{u_T}{2\norm{u_0}}
    \right\|.
    \label{eq:normalized_state_perturbation}
\end{equation}
Thus block error at most $\epsilon/(4u_r)$ is sufficient for
normalized-state error $\epsilon$.

The probability of observing the designated ancilla outcome is
$\Theta(1/u_r^2)$, provided the block error is at most a fixed fraction
of $1/u_r$.  Fixed-point amplitude amplification uses
$\Ocal(u_r)$ applications of the polynomial-transform block and its
inverse to obtain
constant success probability~\cite{Yoder2014fixedpoint}.  Substituting
$\epsilon/u_r$ for the polynomial accuracy in
Eqs.~\eqref{eq:even_degree} and
\eqref{eq:noneven_degree}, or in
Eqs.~\eqref{eq:shifted_integer_degree} and
\eqref{eq:shifted_fractional_degree}, proves
Eqs.~\eqref{eq:even_end_to_end} and
\eqref{eq:noneven_end_to_end}, or
Eqs.~\eqref{eq:shifted_integer_end_to_end} and
\eqref{eq:shifted_fractional_end_to_end}, respectively.

\section{Approximation lower bounds}
\label{app:optimality}

We record the scope of the approximation lower bounds in
Sec.~\ref{subsec:complexity}.

Fix an approximation error smaller than $(1-\mathrm e^{-1})/4$ and
assume $T\norm{H}^{\alpha}$ is larger than an
$\alpha$-dependent constant.  Compare the target at $x=0$ and at
$x=1/[T^{\frac{1}{\alpha}}\norm{H}]$.  Both points then lie in $[-1/2,1/2]$,
and the target values differ by $1-\mathrm e^{-1}$.  Any uniformly
approximating polynomial is bounded by a universal constant and, by
the mean-value theorem, satisfies
$\abs{P'(\xi)}=\Omega(\norm{H}T^{\frac{1}{\alpha}})$ at an intermediate
point.
Bernstein's inequality in the interior,
\begin{equation}
    \abs{P'(\xi)}
    \leq
    \frac{d}{\sqrt{1-\xi^2}}
    \norm{P}_\infty,
    \label{eq:bernstein_derivative}
\end{equation}
therefore gives
\begin{equation}
    d=\Omega(\norm{H}T^{\frac{1}{\alpha}}).
    \label{eq:physical_lower_bound}
\end{equation}
Thus Eq.~\eqref{eq:physical_lower_bound} is an asymptotic lower bound
as $T\norm{H}^{\alpha}\to\infty$.  It matches the physical-scale
dependence of both
Eqs.~\eqref{eq:even_degree} and
\eqref{eq:noneven_degree}.

For even integer $\alpha$, the target is an entire function of order
$\alpha$.  Bernstein's converse theorem for polynomial approximation
of entire functions states that, for fixed $T$ and $H$, the logarithm
of the reciprocal best degree-$d$ error is
$\Theta(d\log d)$.  Hence
$d=\Theta(
\log(1/\epsilon)/\log\log(1/\epsilon))$~\cite{bernstein1912ordre,
timan1963theory}.  For $\alpha=2$, the complete
two-parameter transition, including the logarithmic denominator in
    Eq.~\eqref{eq:even_degree_interpolation}, follows from a known
    optimal-degree theorem~\cite{aggarwal2022optimaldegree}.  For general even powers,
the displayed upper bound, the constant-error lower bound, and the
fixed-$T,H$ high-precision lower bound are the claims used in this
paper; no assertion about sharp universal constants is required.

If $\alpha$ is not even, the even extension used by the direct
construction has the nonzero local term
$-T\norm{H}^{\alpha}\abs{x}^{\alpha}$.  Bernstein's classical
theorem for approximation of $\abs{x}^{\alpha}$, together with the
direct and converse theorems for algebraic singularities, gives best
degree-$d$ error proportional to $d^{-\alpha}$ for this target at
fixed $T$ and $H$~\cite{bernstein1914sur,devore1993constructive,
ditzian1987moduli}.  Therefore
$d=\Theta(\epsilon^{-\frac{1}{\alpha}})$ as
$\epsilon\to0$ within the even-parity class.  This class contains every
gap-independent transformation made by one ordinary QSVT sequence:
the QSVT polynomial has definite parity, and the odd case is excluded
because its value at zero is zero while the target value is one.  This
proves the fixed-scale precision statement in
Table~\ref{tab:degree_regimes}.  Generalized eigenvalue transformations
of indefinite parity are a different constrained-approximation problem
and are not included in this converse statement.  The uniform upper
bound remains Eq.~\eqref{eq:noneven_degree}.

The shifted model has an exact approximation-theoretic reduction.
For every degree $d$,
\begin{align}
\inf_{\deg Q\leq d}
\max_{s\in[-1,1]}
\left|
Q(s)-\mathrm e^{-T\norm{H}^{\alpha}[(1+s)/2]^\alpha}
\right|
=
\inf_{\substack{\deg P\leq2d\\P(-y)=P(y)}}
\max_{y\in[-1,1]}
\left|
P(y)-\mathrm e^{-T\norm{H}^{\alpha}\abs{y}^{2\alpha}}
\right|.
\label{eq:shifted_best_approximation_identity}
\end{align}
Indeed, $Q(s)\mapsto Q(2y^2-1)$ maps degree-$d$ polynomials
bijectively to even polynomials of degree at most $2d$, with inverse
$P(y)\mapsto P(\sqrt{(1+s)/2})$.  Since the target on the right is
even, symmetrizing any approximant shows that the parity restriction
does not change its best uniform error.

If $\alpha$ is a positive integer, the right-hand target in
Eq.~\eqref{eq:shifted_best_approximation_identity} is entire.  The
entire-function converse theorem used above therefore gives the
fixed-scale order
$\Theta(\log(1/\epsilon)/\log\log(1/\epsilon))$ for the
shifted variable as well.  For the physical lower bound, compare the
shifted target at $s=-1$ and at
\begin{equation}
    s=-1+\frac{2}{T^{\frac{1}{\alpha}}\norm{H}}.
\end{equation}
Their values differ by $1-\mathrm e^{-1}$.  The mean-value theorem
requires a derivative of order
$T^{\frac{1}{\alpha}}\norm{H}$, while Markov's inequality gives
$\norm{Q'}_\infty\leq d^2\norm{Q}_\infty$.  Hence
\begin{equation}
    d
    =
    \Omega\!\left(
    (T\norm{H}^{\alpha})^{\frac{1}{2\alpha}}
    \right),
    \label{eq:shifted_physical_lower_bound}
\end{equation}
matching Eq.~\eqref{eq:shifted_integer_degree} at fixed error.

If $\alpha$ is noninteger, the exponent $2\alpha$ is not an even
integer.  Applying the direct and converse theorem for the algebraic
singularity $\abs{y}^{2\alpha}$ to the right-hand side of
Eq.~\eqref{eq:shifted_best_approximation_identity} gives best error
\begin{equation}
    \Theta(d^{-2\alpha}).
    \label{eq:shifted_fractional_best_error}
\end{equation}
Thus the fixed-scale degree is
$\Theta(\epsilon^{-\frac{1}{2\alpha}})$.  The same endpoint
argument leading to Eq.~\eqref{eq:shifted_physical_lower_bound}
applies at fixed error, so
Eq.~\eqref{eq:shifted_fractional_degree} has the optimal physical
scale as well~\cite{devore1993constructive,ditzian1987moduli,
timan1963theory}.

\section{Proof of the Weyl--Poisson identity}
\label{app:weyl_poisson}

We work in finite dimension and set
\begin{equation}
    F(k):=q(k)\mathrm e^{-\mathrm iT(kL+G)},
    \qquad
    \widehat F(\xi):=
    \int_{\mathbb R}F(k)\mathrm e^{-2\pi\mathrm ik\xi}\,\mathrm dk,
    \label{eq:operator_fourier_transform}
\end{equation}
where the Fourier transform is a norm-convergent Bochner integral.  For
real $k$, the operator $kL+G$ is Hermitian.  Duhamel's formula gives
\begin{equation}
\frac{\mathrm d}{\mathrm dk}
\mathrm e^{-\mathrm iT(kL+G)}
=-\mathrm iT\int_0^1
\mathrm e^{-\mathrm iT(1-s)(kL+G)}L
\mathrm e^{-\mathrm iTs(kL+G)}\,\mathrm ds,
\label{eq:weyl_duhamel_derivative}
\end{equation}
so its derivative has norm at most $T\norm{L}$.  Repeated Duhamel
formulas bound every higher derivative uniformly on the real axis.
The Gaussian factor in $q$ therefore makes $F$ a matrix-valued Schwartz
function.  Consequently both lattice sums below converge absolutely in
operator norm, and the usual Poisson argument applies without a
distributional interpretation.

For completeness, the $\ell$th Fourier coefficient of the
$h$-periodization of $F$ is
\begin{align}
\frac{1}{h}\int_0^h
\sum_{m\in\mathbb Z}F(x+mh)
\mathrm e^{-\frac{2\pi\mathrm i\ell x}{h}}\,\mathrm dx
=
\frac{1}{h}\int_{\mathbb R}F(k)
\mathrm e^{-\frac{2\pi\mathrm i\ell k}{h}}\,\mathrm dk
=\frac{1}{h}\widehat F(\ell/h).
\label{eq:operator_periodization_coefficient}
\end{align}
The first equality follows by unfolding the periodized integral; the
phase is unchanged because $\ell m$ is an integer.  Since the Fourier
series is absolutely convergent, evaluating it at zero and multiplying
by $h$ gives
\begin{equation}
    h\sum_{m\in\mathbb Z}F(mh)
    =
    \sum_{\ell\in\mathbb Z}\widehat F(\ell/h).
    \label{eq:operator_poisson_full}
\end{equation}
For the Fourier--Stieltjes Weyl calculus restricted to measures on the
line $(k,1)$~\cite{weyl1950theory,anderson1969weyl,
jefferies2004spectral,ni2026quantumeigenvaluetransformationlinear},
Eq.~\eqref{eq:weyl_fourier_class} shows directly that
\begin{equation}
\widehat F(\ell/h)
=
\operatorname{Op}^{\mathrm W}_{TL,TG}
\!\left[
\mathrm e^{-\mathrm iy}
\phi\!\left(x+\frac{2\pi\ell}{h}\right)
\right].
\label{eq:weyl_mode_identification}
\end{equation}
Substitution into Eq.~\eqref{eq:operator_poisson_full} and separation of
the zero mode prove Eq.~\eqref{eq:weyl_poisson_full}.  The zero mode is
the continuous LCHS integral.  Removing the lattice terms
$\abs{m}>N$ and subtracting the target proves
Eq.~\eqref{eq:weyl_poisson_exact}.  At no point is
$\mathrm e^{-\mathrm iT(kL+G)}$ replaced by a product of exponentials,
so the proof does not require $[L,G]=0$.

Absolute convergence establishes the identity.  To bound the alias sum
quantitatively, one may additionally use strip analyticity.
Suppose $F$ is analytic for $\abs{\operatorname{Im}z}<b$, tends to zero
uniformly on narrower horizontal strips, and its boundary norms are
integrable.  The standard rectangular-contour proof of the
trapezoidal rule, valid for Banach-valued analytic functions, gives
\begin{align}
\left\|
h\sum_{m\in\mathbb Z}F(mh)-\int_{\mathbb R}F(k)\,\mathrm dk
\right\|\leq
\frac{
\displaystyle\int_{\mathbb R}\norm{F(x+\mathrm ib)}\,\mathrm dx+
\displaystyle\int_{\mathbb R}\norm{F(x-\mathrm ib)}\,\mathrm dx}
{\mathrm e^{\frac{2\pi b}{h}}-1}.
\label{eq:operator_trapezoidal_bound}
\end{align}
This is the operator form of the exponentially convergent
trapezoidal estimate~\cite{trefethen2014exponentially}.

For the kernel in Eq.~\eqref{eq:optimal_lchs_kernel}, the Gaussian
factor supplies real-axis decay, while the poles at $k=\pm\mathrm i$
leave $F$ analytic in every strip $\abs{\operatorname{Im}k}<b<1$.
Since the Hermitian part of
$-\mathrm i[(x+\mathrm iy)L+G]$ is $yL$, the
logarithmic-norm bound gives
\begin{equation}
\norm{\mathrm e^{-\mathrm iT[(x+\mathrm iy)L+G]}}
\leq \mathrm e^{T\abs{y}\norm{L}}.
\label{eq:complex_hamiltonian_bound}
\end{equation}
Taking any fixed $b<1$ in
Eq.~\eqref{eq:operator_trapezoidal_bound} recovers exponential mesh
convergence.  The sharper parameter choice established for the
optimal kernel~\cite{low2025optimalquantumsimulationlinear} gives the
spacing and uniform node count stated in
Eq.~\eqref{eq:uniform_lchs_nodes}.

The frequency-tail term always satisfies
\begin{equation}
    \left\|
    h\sum_{\abs{m}>N}
    q(mh)\mathrm e^{-\mathrm iT(mhL+G)}
    \right\|
    \leq
    h\sum_{\abs{m}>N}\abs{q(mh)},
    \label{eq:weyl_tail_bound}
\end{equation}
because each exponential is unitary.  Bounds for the kernel and alias
terms depend on the selected LCHS kernel.  This is why
Eq.~\eqref{eq:weyl_poisson_exact} separates the three contributions
rather than replacing them by one scalar quadrature error.

It remains to justify the transformed rule.  Set
\begin{equation}
    F_{\sinh}(z):=
    \cosh z\,q(\sinh z)
    \mathrm e^{-\mathrm iT(\sinh zL+G)}.
    \label{eq:sinh_operator_integrand}
\end{equation}
Applying Eq.~\eqref{eq:operator_poisson_full} to $F_{\sinh}$ gives
Eq.~\eqref{eq:sinh_weyl_poisson}.  Its zero Fourier mode is
\begin{equation}
    \int_{\mathbb R}F_{\sinh}(x)\,\mathrm dx
    =
    \int_{\mathbb R}q(k)
    \mathrm e^{-\mathrm iT(kL+G)}\,\mathrm dk,
    \label{eq:sinh_zero_mode}
\end{equation}
where $k=\sinh x$.

More generally, the change of variables $k=\sinh z$ writes the
$\ell$th Fourier mode on the right-hand side of
Eq.~\eqref{eq:sinh_weyl_poisson} exactly as
\begin{align}
\int_{\mathbb R}
\cosh z\,q(\sinh z)
\mathrm e^{-\mathrm iT(\sinh zL+G)}
\mathrm e^{-\frac{2\pi\mathrm i\ell z}{h}}\,\mathrm dz
=
\operatorname{Op}^{\mathrm W}_{TL,TG}\!\left[
\mathrm e^{-\mathrm iy}
\int_{\mathbb R}q(k)\mathrm e^{-\mathrm ikx}
\mathrm e^{-\frac{2\pi\mathrm i\ell}{h}
\operatorname{arsinh}(k)}\,\mathrm dk
\right].
\label{eq:sinh_weyl_mode}
\end{align}
Thus the sinh--sinh map preserves both the Weyl lift and the complete
Poisson alias decomposition; the case $\ell=0$ reduces to
Eq.~\eqref{eq:sinh_zero_mode}.

The transformed kernel has the analytic expression
\begin{equation}
    \cosh z\,q(\sinh z)
    =
    \frac{\mathrm e^{c-\mathrm i c\sinh z}}{\pi\cosh z}
    \exp\!\left[-\frac{\cosh^2z}{4\gamma^2}\right].
    \label{eq:sinh_kernel_analytic}
\end{equation}
It is analytic for $\abs{\operatorname{Im}z}<\pi/2$.  Let
$0<b\leq\pi/8$.  Direct calculation gives
\begin{equation}
\begin{split}
    \abs{\cosh(x\mathbin{\pm}\mathrm ib)}
    &\geq \cos b\cosh x,\\
    \abs{\operatorname{Im}\sinh(x\mathbin{\pm}\mathrm ib)}
    &=\sin b\cosh x,\\
    \operatorname{Re}\cosh^2(x\mathbin{\pm}\mathrm ib)
    &=\cos(2b)\cosh^2x+\sin^2b.
\end{split}
\label{eq:sinh_gaussian_strip}
\end{equation}
Together with Eq.~\eqref{eq:complex_hamiltonian_bound}, these identities
imply
\begin{align}
\norm{F_{\sinh}(x\mathbin{\pm}\mathrm ib)}
\leq
\frac{\mathrm e^c}{\pi\cos b\cosh x}
\exp\!\left[
-\frac{\cos(2b)}{4\gamma^2}\cosh^2x
+(c+T\norm{L})\sin b\cosh x
\right].
\label{eq:sinh_boundary_pointwise}
\end{align}
For positive $a$ and real $v$,
$-ay^2+vy\leq-ay^2/2+v^2/(2a)$.  Applying this inequality with
$y=\cosh x$, and then using
$\int_{\mathbb R}(\cosh x)^{-1}\mathrm dx=\pi$, yields the explicit
boundary estimate
\begin{equation}
\int_{\mathbb R}
\norm{F_{\sinh}(x\mathbin{\pm}\mathrm ib)}\,\mathrm dx
\leq
\frac{\mathrm e^c}{\cos b}
\exp\!\left[
\frac{2\gamma^2(c+T\norm{L})^2\sin^2b}{\cos(2b)}
\right].
\label{eq:sinh_boundary_growth}
\end{equation}

Since $L=(A+A^\dagger)/2$, the access model gives
$\norm{L}\leq\norm{A}\leq\beta_A$.  Choose
\begin{equation}
    b=\min\!\left\{
        \frac{\pi}{8},\frac{1}{2(c+T\beta_A)}
    \right\}.
    \label{eq:sinh_strip_choice}
\end{equation}
Then $(c+T\norm{L})\sin b\leq1/2$ and
$\cos(2b)\geq1/\sqrt2$, so each boundary integral in
Eq.~\eqref{eq:sinh_boundary_growth} is at most
$\mathrm e^{c+\gamma^2}/\cos(\pi/8)$.  It follows directly from
Eq.~\eqref{eq:operator_trapezoidal_bound} that mesh error at most
$\epsilon$ is guaranteed whenever
\begin{equation}
    \frac{2\pi b}{h}
    \geq
    \log\!\left[
        1+\frac{2\mathrm e^{c+\gamma^2}}
        {\epsilon\cos(\pi/8)}
    \right].
    \label{eq:sinh_mesh_condition}
\end{equation}
Since $c$ is fixed and
$\gamma^2=\Ocal(\log(1/\epsilon))$, the physical-scale contribution to
the inverse mesh spacing is proportional to $T\beta_A$.  For fixed
positive $T\beta_A$, or as this scale grows, the resource convention of
Sec.~\ref{sec:poisson_qsvt} therefore gives
\begin{equation}
    h^{-1}
    =
    \Ocal\!\left(
        T\beta_A\log\frac{1}{\epsilon}
    \right).
    \label{eq:sinh_mesh_spacing}
\end{equation}
This is the step at which the block normalization enters the branch
count.  The displayed simplification is not a joint asymptotic formula
as $T\beta_A\to0$; in that limit the additive precision term in
Eq.~\eqref{eq:sinh_mesh_condition} must be retained.

For the lattice tail, Eq.~\eqref{eq:sinh_optimal_kernel_decay} is even
and decreasing on the positive real axis.  Integral comparison and
$u=\sinh x$ therefore give
\begin{align}
    h\sum_{\abs{m}>N}
    \cosh(mh)\abs{q(\sinh(mh))}
    \quad\leq
    \frac{4\mathrm e^c\gamma^2}{\pi\sinh(Nh)}
    \exp\!\left[-\frac{\sinh^2(Nh)}{4\gamma^2}\right].
    \label{eq:sinh_lattice_tail}
\end{align}
Indeed, after the substitution it is enough to use
$\int_a^\infty\mathrm e^{-\frac{u^2}{4\gamma^2}}\,\mathrm du
\leq(2\gamma^2/a)\mathrm e^{-\frac{a^2}{4\gamma^2}}$.
Choose
\begin{equation}
    N=
    \left\lceil
        \frac{\operatorname{arsinh}(R)}{h}
    \right\rceil,
    \qquad R=2c\gamma^2.
    \label{eq:sinh_cutoff_choice}
\end{equation}
Then Eq.~\eqref{eq:sinh_lattice_tail} is $\Ocal(\epsilon)$ after the
constant error allocation described below
Eq.~\eqref{eq:weyl_poisson_exact}.  Also,
$R=\Ocal(\log(1/\epsilon))$ and
$\operatorname{arsinh}(R)=
\Ocal(\log\log(1/\epsilon))$.  Taking $h\leq1$
also gives
$\sinh(Nh)=\Ocal(\log(1/\epsilon))$.  Suppressing the additive constant
number of nodes according to the same resource convention gives
\begin{equation}
    2N+1
    =\Ocal\!\left(
        T\beta_A
        \log\frac{1}{\epsilon}
        \log\log\frac{1}{\epsilon}
    \right),
    \label{eq:sinh_node_count_proof}
\end{equation}
which proves Eq.~\eqref{eq:sinh_lchs_nodes_radius}, including its
$\beta_A$ dependence.

It remains to distinguish this node count from the LCU coefficient
one-norm.  The function
$\cosh x\abs{q(\sinh x)}$ is even and decreasing for $x\geq0$ by
Eq.~\eqref{eq:sinh_optimal_kernel_decay}.  Monotone integral comparison
therefore gives, for $h\leq1$,
\begin{align}
    h\sum_{m=-N}^{N}
    \cosh(mh)\abs{q(\sinh(mh))}
    \leq
    h\sum_{m\in\mathbb Z}
    \cosh(mh)\abs{q(\sinh(mh))}
    \leq
    \int_{\mathbb R}\abs{q(k)}\,\mathrm dk
    +h\abs{q(0)}
    \leq
    \mathrm e^c\left(1+\frac{1}{\pi}\right).
    \label{eq:sinh_discrete_normalization}
\end{align}
Here the substitution $k=\sinh x$ gives the integral in the second
line, Eq.~\eqref{eq:optimal_lchs_normalization} bounds it by
$\mathrm e^c$, and $\abs{q(0)}\leq\mathrm e^c/\pi$.  This proves
Eq.~\eqref{eq:sinh_lchs_normalization_bound}.  The bound contains
neither $T$ nor $\beta_A$: the normalization is constant even though
the number of lattice points is not.

\section{APS factorizations and query complexity}
\label{app:aps}

\subsection{Half-root access and controlled family}

Suppose first that an exact block encoding of
\begin{equation}
    R=\left(\frac{H}{\beta_H}\right)^{\frac{1}{2\alpha}}
    \label{eq:aps_half_root}
\end{equation}
is available.  The Chebyshev polynomial $2x^2-1$ has even parity and
modulus at most one on $[-1,1]$, so a degree-two QSVT sequence implements
it exactly~\cite{Gilyen2019qsvt}.  Since $R$ is positive semidefinite, its singular-value and
eigenvalue transformations coincide, and the designated block is
\begin{equation}
    2R^2-I
    =2\left(\frac{H}{\beta_H}\right)^{\frac{1}{\alpha}}-I.
    \label{eq:aps_half_root_to_shift}
\end{equation}
This proves the relation to the root-shift signal stated in
Eq.~\eqref{eq:aps_root_signal}.

We next verify that all times required by APS can be placed in one
coherent circuit.  For every $t_m\in[0,T]$, let $P_m$ be a polynomial
for $\mathrm e^{-t_m\beta_H\abs{x}^{2\alpha}}$.  Compute the unscaled
Poisson polynomial to uniform error $\delta/4$ and round its coefficients
so that the induced uniform error is at most $\delta/4$.  Dividing the
result by $1+\delta/2$ gives an even polynomial satisfying
\begin{equation}
    \max_{x\in[-1,1]}\abs{P_m(x)}\leq1,
    \qquad
    \left\|
    P_m\!\left((H/\beta_H)^{\frac{1}{2\alpha}}\right)
    -\mathrm e^{-t_mH}
    \right\|\leq\delta.
    \label{eq:aps_normalization_one_family}
\end{equation}
The second bound is Eq.~\eqref{eq:normalization_two_scalar_bound} with
$\epsilon$ replaced by $\delta$, but without the additional fixed factor
$1/2$: before rescaling, the rounded polynomial is within $\delta/2$ of
the target and has modulus at most $1+\delta/2$.
Thus the real-polynomial QSVT theorem for fixed parity implements $P_m$
as a normalization-one approximate block encoding
~\cite{Gilyen2019qsvt}, rather than the fixed normalization-two block
used in Theorem~\ref{thm:poisson_compilation}.  Choose a common even
query length no smaller than the maximum required at $T$.  Shorter
sequences may be padded by cancelling pairs of signal queries.  A QSVT
sequence is an
alternating product of the signal unitary and phase rotations.  At the
$j$th rotation, replace the phase gate by the direct sum of the phase
gates for all $m$, controlled by the time register.  The signal unitary
itself is independent of $m$.
Consequently the full circuit is block diagonal in that register, and
its designated block obeys
\begin{align}
\left\|
\sum_m\ket m\!\bra m\otimes
\left[
P_m\!\left(
\left(\frac{H}{\beta_H}\right)^{\frac{1}{2\alpha}}
\right)
-\mathrm e^{-t_mH}
\right]
\right\|
=\max_m\left\|
P_m\!\left(
\left(\frac{H}{\beta_H}\right)^{\frac{1}{2\alpha}}
\right)
-\mathrm e^{-t_mH}
\right\|
\leq\delta.
\label{eq:aps_controlled_family_error}
\end{align}
This gives the Poisson-compiled controlled dissipative-family oracle
used below.  It requires the maximum polynomial degree,
not the sum of the degrees over the time grid~\cite{Gilyen2019qsvt}.
This is a signal-query statement.  The gate cost of the multiplexed
phase rotations additionally depends on how the time register is
accessed and is not included in the block-encoding query count.

\subsection{Exact factorizations}

Let
$U(T)=\mathrm e^{-T(H+\mathrm iG)}$.
Writing $U(T)=\mathrm e^{-\mathrm iTG}V(T)$ and differentiating gives
\begin{equation}
    V'(T)
    =
    -\mathrm e^{\mathrm iTG}
    H
    \mathrm e^{-\mathrm iTG}V(T),
    \qquad V(0)=I.
    \label{eq:phase_picture_ode}
\end{equation}
The time-ordered solution is
Eq.~\eqref{eq:phase_driven_aps}.
Alternatively, writing
$U(T)=\mathrm e^{-T H}W(T)$ gives
\begin{equation}
    W'(T)
    =
    -\mathrm i\,
    \mathrm e^{T H}
    G
    \mathrm e^{-T H}W(T),
    \qquad W(0)=I,
    \label{eq:amplitude_picture_ode}
\end{equation}
which proves Eq.~\eqref{eq:amplitude_driven_aps}.  These identities are
exact.  Algorithmic use of either one additionally requires the
controlled access and time-ordering assumptions specified by the APS
implementation.

\subsection{Dissipative Dyson products}

The amplitude-driven factorization can be used without separately
implementing the possibly non-Hermitian similarity transform
$\mathrm e^{sH}G\mathrm e^{-sH}$.  Iterating
Eq.~\eqref{eq:amplitude_picture_ode} and multiplying by
$\mathrm e^{-TH}$ from the left gives the exact series
\begin{align}
\mathrm e^{-T(H+\mathrm iG)}
=\sum_{k=0}^{\infty}(-\mathrm i)^k
\int_{0\leq s_1\leq\cdots\leq s_k\leq T}
\mathrm e^{-H(T-s_k)}G
\times
\mathrm e^{-H(s_k-s_{k-1})}G\cdots
G\mathrm e^{-Hs_1}\,
\mathrm ds_1\cdots\mathrm ds_k .
\label{eq:aps_dissipative_dyson}
\end{align}
For $k=0$ the integrand is $\mathrm e^{-TH}$, while for $k=1$ the
product is
$\mathrm e^{-H(T-s_1)}G\mathrm e^{-Hs_1}$.  Every dissipative
duration in a term of Eq.~\eqref{eq:aps_dissipative_dyson} is
nonnegative, and the durations sum to $T$.  Since
$H\succeq0$, all of these factors are contractions.  The norm of the
$k$th term is consequently at most
\begin{equation}
    \frac{(\beta_GT)^k}{k!}.
    \label{eq:aps_dyson_term_bound}
\end{equation}

The APS implementation divides the evolution into
$\Ocal(\beta_GT)$ segments and
truncates the local Dyson series at
\begin{equation}
    \Ocal\!\left(
    \frac{\log(u_r/\epsilon)}
    {\log\log(u_r/\epsilon)}
    \right)
    \label{eq:aps_local_dyson_order}
\end{equation}
terms~\cite{hu2026quantumsimulationnonunitarydynamics}.  Thus the total
number of alternating dissipative and $G$ blocks is
$\Ocal(\beta_G T\log(u_r/\epsilon)/\log\log(u_r/\epsilon))$.  After the ordered integrals are
discretized, every occurrence of
$\mathrm e^{-H(s_j-s_{j-1})}$ is selected from
Eq.~\eqref{eq:aps_controlled_family}; Eq.~\eqref{eq:aps_controlled_family_error}
therefore supplies all dissipative blocks in the APS circuit with one
uniform construction.  The Poisson-compiled polynomial thereby
provides the controlled dissipative-family realization used below.

\subsection{Error and query allocation}

We now justify Corollary~\ref{cor:aps_complexity}.  The APS circuit makes
$\Ocal(\beta_G T\log(u_r/\epsilon)/\log\log(u_r/\epsilon))$
controlled dissipative-family calls per run.  Telescoping over those
calls and the subsequent amplitude-amplification steps shows that it is
sufficient to choose the individual block error as
\begin{equation}
    \Theta\!\left(
    \frac{\epsilon}{u_r\beta_G T}
    \frac{\log\log(u_r/\epsilon)}
    {\log(u_r/\epsilon)}
    \right)
    \label{eq:aps_individual_accuracy}
\end{equation}
and to allocate the same order of total error to the truncated Dyson
construction.  Eq.~\eqref{eq:normalization_lemma} then controls
normalization of the approximate output.

Every time in Eq.~\eqref{eq:aps_controlled_family} lies in $[0,T]$.
The uniform degree bounds in Theorem~\ref{thm:degree} are monotone when
each selected time is upper-bounded by $T$.  Thus all
controlled branches may be padded to the maximum degree required at
$T$.  The multiplexed construction in
Eq.~\eqref{eq:aps_controlled_family_error} supplies their coherent
programmability.
To apply Theorem~\ref{thm:degree}, regard $H$ as
$(H^{\frac{1}{2\alpha}})^{2\alpha}$.  The half-root oracle in
Eq.~\eqref{eq:aps_half_root_access} supplies the standard positive signal
for $H^{\frac{1}{2\alpha}}$ with normalization
$\beta_H^{\frac{1}{2\alpha}}$.  Substituting this normalization, the
even exponent $2\alpha$, and the polynomial tolerance in
Eq.~\eqref{eq:aps_individual_accuracy} into Eq.~\eqref{eq:even_degree}
gives
\begin{equation}
\Ocal\!\left(
\left(
T\beta_H+\log\frac{u_r\beta_G T}{\epsilon}
\right)^{\frac{1}{2\alpha}}
\log^{1-\frac{1}{2\alpha}}\!\frac{u_r\beta_G T}{\epsilon}
\right).
\label{eq:aps_dissipative_resource}
\end{equation}
This is the number of half-root queries executed by one invocation of
the controlled dissipative family at the required uniform accuracy.  It
is Eq.~\eqref{eq:aps_family_degree}.  For nonnegative $a,b$ and
$0<p\leq1$,
concavity gives
$(a+b)^p\leq a^p+b^p\leq2^{1-p}(a+b)^p$.
Consequently Eq.~\eqref{eq:aps_dissipative_resource} is, up to a
constant depending only on $\alpha$, the same uniform bound as the
physical-scale term plus $\log(u_r\beta_G T/\epsilon)$.  If
$T\beta_H\geq\log(u_r\beta_G T/\epsilon)$, its first factor is bounded
by a constant times
$(T\beta_H)^{\frac{1}{2\alpha}}
\log^{1-\frac{1}{2\alpha}}(u_r\beta_G T/\epsilon)$.
For $\alpha=1$, this is precisely the dissipative term in
Eq.~\eqref{eq:aps_square_root_reduction}.

The extra factor
$\log(u_r/\epsilon)/\log\log(u_r/\epsilon)$ inside the reciprocal
individual accuracy changes its logarithm only by a lower-order
iterated-log term, which is absorbed in the displayed order.  The
number of controlled-family calls per run and the number of calls to
$G/\beta_G$ have the APS order
$\Ocal(\beta_G T\log(u_r/\epsilon)/\log\log(u_r/\epsilon))$
~\cite{hu2026quantumsimulationnonunitarydynamics}.  Fixed-point
amplitude amplification multiplies both counts by $u_r$
~\cite{Yoder2014fixedpoint}, which gives
Eq.~\eqref{eq:aps_family_calls} and
Eq.~\eqref{eq:aps_flat_g_queries}.

If Eq.~\eqref{eq:aps_controlled_family} is supplied directly,
Eq.~\eqref{eq:aps_family_calls} is its query count and
Eq.~\eqref{eq:aps_family_degree} specifies one possible realization.
When the family is implemented by that QSVT circuit, each invocation
executes the half-root sequence.  Substituting
Eq.~\eqref{eq:aps_dissipative_resource} into the family-call count of
the amplitude-driven APS framework~\cite{hu2026quantumsimulationnonunitarydynamics} gives
Eq.~\eqref{eq:aps_flat_half_root_queries}.  The same procedure
uses $\Ocal(u_r)$ calls to the initial-state oracle.  This proves
Corollary~\ref{cor:aps_complexity}.
The sharper interpolation obtained from
Eq.~\eqref{eq:even_degree_interpolation} after replacing the exponent by
$2\alpha$ and the normalization by $\beta_H^{\frac{1}{2\alpha}}$
may be substituted in this resource bound; equivalently, it is
Eq.~\eqref{eq:shifted_integer_degree_interpolation}.  At fixed
$T\beta_H$ and $\beta_G T$ it
improves the per-invocation high-precision behavior to
$\log(u_r\beta_G T/\epsilon)/\log\log(u_r\beta_G T/\epsilon)$, as stated
after the corollary.

\end{document}